\documentclass{article}
\usepackage{float}
\usepackage[margin=1in]{geometry}
\usepackage{titlesec}
\titleformat{\chapter}[display]
{\normalfont\huge\bfseries}{}{0pt}{\Huge}
\usepackage{graphicx}
\usepackage[utf8]{inputenc}
\usepackage{algorithm}
\usepackage{algpseudocode}
\usepackage{amscd} 
\usepackage{amsmath} 
\usepackage{amssymb} 
\usepackage{amsthm} 
\usepackage[style=numeric-comp,sorting=nyt,sortcites=true]{biblatex}

\usepackage{comment}
\usepackage{appendix}
\usepackage{amsthm} 
\usepackage{dsfont}
\usepackage{graphicx}
\usepackage{subcaption}
\usepackage{xcolor}
\usepackage{soul}
\usepackage{wrapfig}
\usepackage{caption}
\usepackage{lipsum}
\usepackage{float}

\usepackage{algorithm}
\usepackage{algpseudocode}
\usepackage{enumitem}

\newtheorem{lemma}{Lemma}
\newtheorem{proposition}{Proposition}
\newtheorem{theorem}{Theorem}
\newtheorem{remark}{Remark}

\newtheorem{prop}{Property}

\theoremstyle{definition}
\newtheorem{definition}{Definition}
\usepackage[skip=5pt plus1pt, indent=0pt]{parskip}

\newcommand{\OO}{\mathcal{O}}
\DeclareMathOperator{\E}{\mathbb E}

\DeclareMathOperator{\T}{\mathcal T}
\DeclareMathOperator{\R}{\mathds R}

\DeclareMathOperator{\WT}{\mathcal W_{\mathcal T}}

\newcommand{\RNum}[1]{\uppercase\expandafter{\romannumeral #1\relax}}

\DeclareMathOperator*{\argmin}{arg\,min}
\newcommand{\bl}[1]{\textcolor{blue}{#1}}

\usepackage{authblk}

\title{SDSR: A Spectral Divide-and-Conquer Approach for Species Tree Reconstruction}
\date{}
\author[1]{Ortal Reshef }
\author[3]{Ofer Glassman}
\author[1]{Or Zuk}
\author[2]{Yariv Aizenbud}
\author[3]{Boaz Nadler}
\author[1]{Ariel Jaffe}

\affil[1]{\small Department of Statistics and Data Science, Hebrew University of Jerusalem}
\affil[2]{\small Department of Applied Mathematics, Tel Aviv University}
\affil[3]{Department of Computer Science and Applied Mathematics, Weizmann Institute of Science}
\begin{document}
\maketitle

\begin{abstract}
Recovering a tree 
that represents the evolutionary history of a group of species is a key task in phylogenetics.
Performing this task using sequence data from multiple genetic markers poses two key challenges. The first is the 
discordance between
the evolutionary history of individual genes
and that of the species. 
The second challenge is computational, as 
contemporary studies involve thousands of species. 
Here we present \texttt{SDSR}, a scalable divide-and-conquer approach for species tree reconstruction based on spectral graph theory.
The algorithm recursively partitions the species into subsets until their sizes are below a given threshold. 
The trees of these subsets are reconstructed by a user-chosen
species tree algorithm. Finally, these subtrees are merged
to form the full tree. 
On the theoretical front, 
we derive recovery guarantees for \texttt{SDSR},
under the multispecies coalescent (MSC) model. 
We also perform a runtime complexity analysis. We show that \texttt{SDSR}, when combined with a species tree reconstruction algorithm as a subroutine, yields substantial runtime savings as compared to applying the same algorithm on the full data. 
Empirically, we evaluate \texttt{SDSR} on synthetic benchmark datasets with incomplete lineage sorting and horizontal gene transfer.
In accordance with our theoretical analysis, the simulations show that combining \texttt{SDSR} with common species tree
methods, such as \texttt{CA-ML} or
\texttt{ASTRAL}, yields up to 10-fold faster runtimes.
In addition, \texttt{SDSR} achieves a comparable tree reconstruction accuracy to that obtained by applying these methods on the full data. 

\end{abstract}

\section{Introduction}

Species tree reconstruction from genetic data is crucial for addressing many evolutionary questions, such as diversification events, branching structures, and shared ancestry among organisms \cite{YangRannala2012, wu2025treehub}. 
Historically, species trees were inferred using genetic variations at a single marker. 
Common approaches for tree recovery include distance-based methods such as \texttt{NJ} and \texttt{UPGMA} \cite{saitou1987neighbor, rr1958statiscal}, maximum likelihood (ML) \cite{felsenstein1981evolutionary, stamatakis2014raxml, nguyen2015iq, price2009fasttree}, and Bayesian methods \cite{rannala1996probability, yang1997bayesian, drummond2007beast}.

In recent years, advances in data acquisition
have led to species tree reconstruction using sequences from multiple genetic markers
\cite{liu2011estimating, allman2016species, vachaspati2015astrid, zhang2018astral}.
Species tree reconstruction from such data
introduces two key challenges. 
The first is the discordance between the evolution of the species and of their genes, resulting in gene trees that may be different from each other, and from the species tree. This discordance arises from several biological phenomena \cite{cotton2005rates,
soucy2015horizontal,
barton2001role, paquola2018}, with
horizontal gene transfer (HGT) and incomplete lineage sorting (ILS)
considered as two of the most significant contributors
\cite{tandy1, tan2023phylogenomics, maddison2006inferring},
see Fig.
\ref{fig:ils_hgt} for an illustration. 
A second challenge is scalability, as contemporary phylogenetic studies often analyze tens of thousands of species and hundreds of genes \cite{baker2022comprehensive, one2019one, zhu2019phylogenomics}.

\begin{figure}[t]
    \centering   \fbox{\includegraphics[width=0.85\textwidth]{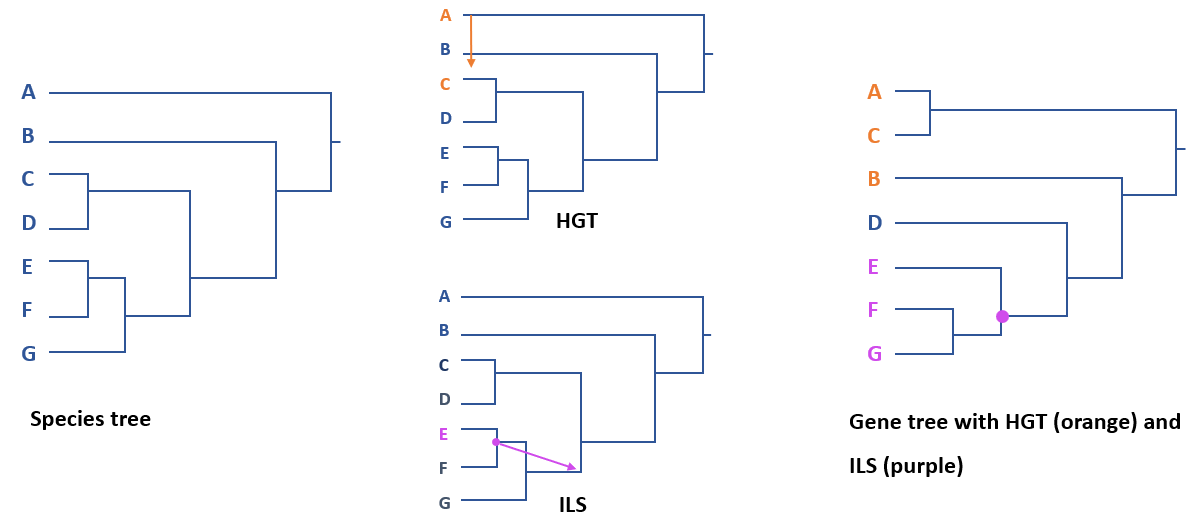}}
    \caption{ 
    Discordance
    between a species tree and a gene tree due to one HGT event (orange) and one ILS event
    (purple). In HGT, a  segment of genetic material is transferred horizontally from a donor species $A$, to a receptor species $C$ \cite{hall2017sampling, brito2021examining}.
    In ILS, the divergence event associated with species $E$ occur at an earlier time in the gene tree than in the species tree \cite{degnan2009gene}.
    }
    \label{fig:ils_hgt}
\end{figure}

Several approaches were proposed for species tree reconstruction from sequences at multiple markers. Perhaps one of the simplest ones is 
Concatenation Analysis using Maximum Likelihood (\texttt{CA-ML}). In this approach, the gene sequences of each species are concatenated,  
and the species tree is recovered by maximum likelihood. 
Hence, \texttt{CA-ML} ignores the discordance between gene and species trees. 
Indeed, \texttt{CA-ML} has been shown to be statistically inconsistent under the multispecies coalescent (MSC) model 
\cite{kubatko2007inconsistency,roch2015likelihood}.

A different approach to species tree reconstruction is to first estimate all the individual gene trees by some tree reconstruction method. Then, 
the species tree is inferred by summarizing information from all gene trees.
Methods applying this approach are known as 
{\em summary methods}. Examples include 
\texttt{NJst} \cite{liu2011estimating}, 
\texttt{STAR} \cite{liu2009estimating}, 
\texttt{ASTRID} \cite{vachaspati2015astrid}, 
and 
\texttt{ASTRAL} \cite{mirarab2014astral,mirarab2015astral,zhang2018astral}. 
These summary methods were proven to be consistent under the 
MSC model, assuming the input gene trees are estimated correctly \cite{allman2016species, liu2009estimating, vachaspati2015astrid, mirarab2014astral}.
However, errors in the inferred gene trees may have a substantial impact on the accuracy of the reconstructed species tree 
see 
\cite{roch2015robustness} and \cite[Ch. 10]{warnow2018computational}.
On the computational front, summary methods may be slow on large datasets, as they first reconstruct all the individual gene trees.

To handle datasets with thousands of species, several divide-and-conquer algorithms
were developed. 
These methods break down the species tree reconstruction into smaller, more manageable tasks. 
One strategy
is to randomly divide the species into overlapping subsets and construct a subtree for each subset using an existing species tree inference method. The resulting subtrees are then merged into a species tree using a supertree algorithm \cite{huson1999disk, huson1999solving, nelesen2012dactal}.
This merging operation involves 
NP-hard optimization problems \cite{warnow2019divide}, which pose
limitations on the accuracy and scalability
of supertree algorithms. 
\texttt{TreeMerge} \cite{molloy2018njmerge, molloy2019treemerge} partitions the species into disjoint subsets.
Then, it reconstructs a subtree for each and merges them using a species dissimilarity matrix.
\texttt{uDance} \cite{balaban2024generation} offers an alternative divide-and-conquer strategy.
It recovers a backbone tree from a subset of species and assigns each of the remaining species to an edge in the backbone tree. Then, the algorithm reconstructs a subtree for every group of species assigned to the same edge.
Empirically, on simulated datasets with a large number of genes and species, \texttt{uDance} outputs quite accurate species trees. 
However, \texttt{uDance} has not been proven to be statistically consistent under the MSC or other models. 

In this paper we present
\texttt{SDSR}, a recursive divide-and-conquer method
to recover a species tree from multiple gene sequences. 
At each step of the recursion, 
if the input set of species is sufficiently small, its corresponding
species subtree is reconstructed by a user-chosen species tree method.
Otherwise, the set of species is partitioned into two disjoint subsets. This is performed by a spectral approach, based on 
the Fiedler eigenvector of the average of graph Laplacian matrices for the different genes.  
Finally, the small trees are merged to construct the full species tree.
The partitioning step builds upon and extends \cite{aizenbud2021spectral}, which was designed to reconstruct a tree from a single gene alignment. 

With \texttt{SDSR}, we make the following contributions:
\begin{itemize}[leftmargin=*]
    \item 
    We prove that under the MSC model, 
    in the limit of an infinite number of genes, 
    \texttt{SDSR} is asymptotically consistent. 
    Specifically, the species in each partition belong
    to disjoint clans in the species tree. In addition, we derive finite-sample guarantees on the number of genes required for accurate partitions. Our analysis assumes that the Laplacian matrices of individual genes are perfectly known. This is analogous to the assumption made in the analysis of summary methods, that their gene trees are inferred without errors. 
    \item The \texttt{SDSR} partitions are deterministic, which yields a simple and straightforward way to merge the subtrees into the full species tree. Moreover, if \texttt{SDSR} partitioned correctly the species and the subtrees were recovered correctly, the merging step is provably correct, so the output will be the exact species tree. Importantly,  \texttt{SDSR} merging step does not require solving an  NP-hard problem, unlike previous divide-and-conquer methods.
    \item \texttt{SDSR} combined with common species tree reconstruction algorithms, such as \texttt{CA-ML} or
    \texttt{ASTRAL} yields significant speedups compared to applying the base algorithm to reconstruct the full tree. For example, reconstructing a 200-species tree from 100 genes
    with \texttt{SDSR + CA-ML} has an $8$-times faster runtime
    on a single CPU, while achieving accuracy similar to \texttt{CA-ML} on the full data. 
    In addition, \texttt{SDSR} can be trivially parallelized. 
\end{itemize}

The rest of the paper is organized as follows.
In Section \ref{sec:background}, we introduce the problem setting and several definitions. 
The \texttt{SDSR} algorithm is described in Section \ref{sec:method}.
Its computational complexity is analyzed in Section \ref{sec:complexity}. 
Theoretical guarantees under the MSC model are provided in Section \ref{sec:analysis}.
In Section \ref{sec:experiments}, we 
present an empirical evaluation 
of \texttt{SDSR}, in terms of accuracy and runtime, and compare it to several competing species tree methods.
We conclude with a summary and discussion in Section \ref{sec:discussion}. 

\section{Problem Setting}\label{sec:background}

Let $\T = (V, E)$ denote an unrooted binary species tree with $m$ terminal and $m-2$ non-terminal nodes. Terminal nodes correspond to present-day species, whereas non-terminal nodes to ancestor species. 
We denote the indices of the terminal nodes by $[m] \equiv \{1,2,\ldots,m\}$ and the non-terminal nodes by $\{m+1,\ldots,2m-2\}$. 

The observed data for each current species at a specific gene is a sequence of symbols from a set $\{\alpha_1,\ldots,\alpha_\ell\}$. In the common DNA setting, $\ell = 4$ and the set is $\{A,C,G,T\}$ which are the four 
nucleotide bases.
In the problem at hand, we are given a set of $K$ gene alignments, denoted $\{X^g\}_{g=1}^K$. A single gene alignment $X^g \in \{\alpha_1,\ldots,\alpha_\ell\}^{m \times {n_g}}$ contains $m$ sequences of length $n_g$, where each sequence corresponds to a single species. 
We denote by $X_i^g = (X_{i1}^g,\ldots,X_{in_g}^g)$  the sequence that corresponds to species $i$ in the $g$-th gene. 
The goal is to reconstruct the species tree from these $K$ gene alignments. 

In this work, we present a spectral divide-and-conquer approach for species tree reconstruction. Our algorithm is based on a weighted graph, where the nodes correspond to the species, and the weights 
measure a specific similarity between the species. As detailed below this similarity is computed using all $K$  gene alignments. 
To this end, we introduce several notations. 

\paragraph{Similarity measure.}
For each gene $g$ with gene alignment matrix 
$X^g$,
we compute a measure of evolutionary distance $\hat D^g(i,j)$ between all pairs of observed species $(i,j)$. For example, a common measure is the log-determinant distance, which estimates the determinants of the transition matrices between $X_i^{g}$ and $X_j^{g}$  \cite{barry1987asynchronous}. Other measures take into account varying mutation rates along the sequence \cite{yang1994maximum}. 
Next, we compute the gene $m\times m$ similarity matrix between all $m$ species, via 
\begin{equation}\label{eq:similarity}
\hat S^g(i,j) = \exp\big(-\hat D^g(i,j)\big), \qquad \forall i,j \in [m]. 
\end{equation}

\paragraph{Graph Laplacian and Fiedler vector.}
The above similarity matrix $\hat S^g$ can be viewed as a weight matrix of an undirected graph $G^{g}$, whose nodes are the $m$ observed species. 
Let $A^{g}$ denote the diagonal matrix with entries $A^{g}(i,i) = \sum_{j=1}^m \hat S^g(i,j)$.
The normalized Laplacian $L^{g}$ of the similarity graph is defined by,
\begin{equation}\label{eq:laplacians}
 L^{g} = I - (A^{g})^{-0.5} \hat S^g (A^{g})^{-0.5}. %\qquad L_s = I - D^{-0.5}WD^{-0.5}
\end{equation}
For any connected graph, the Laplacian matrix is positive semi-definite with a zero eigenvalue with multiplicity one. The Fiedler vector is the eigenvector of $L$ that corresponds to its smallest non-zero eigenvalue.

Finally, to 
motivate our approach,  we recall some standard definitions for trees.

\paragraph{A clan and its root.}  
A subset of nodes $C$ in a tree $\T$ is a {\em clan} if it can be separated from the rest of the tree by removing a single edge $e$.  
The node $h_C \in C$ that is at one end of the edge $e$ is defined as the root of the smallest subtree that includes $C$.
In our work, we often refer to a clan by its terminal nodes. For example, for the species tree in Figure \ref{fig:ils_hgt}, the sets of nodes $(E, F)$ and $(E, F, G)$ form clans in $\T$. 

\section{The \texttt{SDSR} algorithm}\label{sec:method}

In this section we present our divide-and-conquer approach for recovering  species trees. 
The input to \texttt{SDSR} consists of a set of $K$ gene alignments $\{X^{g}\}_{g=1}^K$, a threshold $\tau$, and a user-chosen 
species tree reconstruction
method to recover the subtrees, which we refer to as a subroutine. The algorithm is recursive, and consists of four basic steps: 
(i) Partition the terminal nodes into two disjoint subsets $C_1,C_2$. (ii) Add 
to $C_2$ an outgroup node $O_1\in C_1$ and vice versa. Let $\tilde C_1 = C_1 \cup \{O_2\}$ and $\tilde C_2 = C_2 \cup \{O_1\}$.
(iii) For each $i=1,2$, reconstruct the tree of $\tilde C_i$ as follows: if $|\tilde C_i|>\tau$, recursively apply  \texttt{SDSR}. 
Otherwise, recover its species tree by the subroutine. (iv) Merge the two subtrees computed in (iii). A pseudocode for \texttt{SDSR} is outlined in Alg. \ref{alg:SDSR}.
We next describe each of the above steps, see illustration in Figure \ref{fig:merging}.

\begin{algorithm}[t]
\caption{\texttt{SDSR}}\label{alg:cap}
\label{alg:SDSR}
\begin{algorithmic}%[1]
\Require Gene alignments $\{X^{g}\}_{g=1}^K$ for a set of species $C$, threshold $\tau$, tree reconstruction subroutine, minimum partition ratio $\beta$.
\Ensure Reconstructed species tree.
\State For every gene, compute an estimated distance matrix $\hat D^{g}$, its similarity matrix  $S^g$ via Eq. \eqref{eq:similarity}, and the Laplacian $L^{g}$ by Eq. \eqref{eq:laplacians}. Then, compute the average Laplacian $\bar L$ via Eq. \eqref{eq:average_laplacian}.
\State Compute a partition $C_1,C_2$ of the set of species $C$ by constrained $k$-means with minimum size $\beta |C|$ on the Fiedler vector of $\bar L$.
\State Select outgroups $O_1 \in C_1$ and $O_2 \in C_2$ and set $\tilde C_1 = C_1 \cup \{O_2\}$ and $\tilde C_2 = C_2 \cup \{O_1\}$.
\For{i = 1,2} 
\If{$|\tilde C_i|>\tau$}
    \State Apply \texttt{SDSR} to the subset of species $\tilde C_i$.    
\Else
    \State Recover subtree $\tilde \T_i$ by applying the tree reconstruction subroutine to species $\tilde C_i$.  
\EndIf
\EndFor
\State Let $h_1$ and $h_2$ be nodes in
$\tilde \T_1$ and $\tilde T_2$ connected to
$O_2$ and to $O_1$, respectively. 
\State Remove the outgroups $O_2,O_1$ from 
$\tilde \T_1$, $\tilde T_2$. 
\State Compute a merged tree $\hat\T$ by adding an edge between $h_1,h_2$ 
in the two subtrees $\tilde \T_1,\tilde \T_2$. 
\State \Return merged tree $\hat\T$.
\end{algorithmic}
\end{algorithm}

\paragraph{Step 1: Partitioning.}
For each gene $g$, let $\hat S^g(i,j)$ be the similarity between the two species $i$ and $j$, computed from their gene alignments via 
Eq. \eqref{eq:similarity}. 
We compute a gene-Laplacian matrix $L^{g}$ via Eq. \eqref{eq:laplacians}, and the averaged Laplacian, denoted $\bar L$, by
\begin{equation}\label{eq:average_laplacian}
\bar L = \frac1K \sum_{g=1}^{K} L^{g}. 
\end{equation}
We partition the species into subsets $C_1,C_2$ via their corresponding elements in the Fiedler vector of $\bar L$. 
In order to avoid extremely imbalanced partitions that do not reduce the computational complexity of the recovery problem, we apply constrained k-means \cite{bradley2000constrained} to the elements of the Fiedler vector. Specifically, for an input set of species $C$ we add a constraint that the smaller subset is at least $\beta |C|$, for some value $0<\beta<0.5$. 

\paragraph{Step 2: Adding outgroup nodes.} 
As described below in step 4, the merging step relies on \textit{outgroup rooting} \cite{grant2019outgroup, boykin2010comparison}, 
a common approach to find the root of unrooted trees. The main idea is to reconstruct the tree with an additional species, related to but outside of the species group of interest. 
This added species serves as a reference point for the root's location. 

In our method,  outgroups are used to merge the subtrees. To this end, we add an \textit{outgroup species} $O_2 \in C_2$ to the \textit{ingroup species} $C_1$, and vice versa. 
Let $\tilde{C}_1 = C_1 \cup \{O_2\}$ and $\tilde{C}_2 = C_2 \cup \{O_1\}$ be the  resulting subsets.  
Several works showed that the selection of an outgroup may significantly impact the accuracy of the estimated root. Specifically, the distance of the chosen outgroup to the ingroup species should not be too large or too small \cite{desalle2023multiple, bergsten2005review}. Here, we select the outgroup according to the following simple procedure. 
Let $v\in \R^m$ denote the Fiedler vector of $\bar L$ computed in Eq. \eqref{eq:average_laplacian}. We view distances in the Fiedler vector $|v_i-v_j|$ as indicative of distances between species $i$ and $j$.
We choose $O_2 \in C_2$ by looking for the species whose average distance to the species in $C_1$ is closest to the average of all distances from $C_2$ to $C_1$. Namely, 
for each node $j \in C_2$ define $d_j = \frac{1}{|C_1|}\sum_{i \in C_1}  |v_j - v_i|$ as the average distance from $j$ to the species in $C_1$ and by $D(C_1,C_2) = \frac{1}{|C_2|} \sum_{j \in C_2} d_j$.  
Let $\bar v^{(2)} = \frac{1}{|C_2|}\sum_{i \in C_2}^m v_i$ denote the average of the elements that correspond to $C_2$. By simple linear algebra
\[
j = \argmin_{j' \in C_2} |d_{j'} - D(C_1,C_2)| = \argmin_{j' \in C_2} |v_{j'}-\bar v^{(2)}|.
\]
The outgroup $O_1 \in C_1$ is selected in a similar way. 

\paragraph{Step 3: Reconstruction.}
For each of the two subgroups $\tilde C_1$
and $\tilde C_2$, 
if $|\tilde C_i|>\tau$, 
recursively call \texttt{SDSR} with input set
$\tilde C_i$ for further partitioning. 
Otherwise, apply a pre-selected reconstruction algorithm such as 
\texttt{ASTRAL} or \texttt{NJst}, to recover
the (unrooted) species tree of $\tilde C_i$. 

\paragraph{Step 4: Merging the two subtrees.} 
The last step of \texttt{SDSR} is to merge the unrooted subtrees $\tilde \T_1,\tilde \T_2$ into a single tree. If the  partitioning and recovery steps were accurate, then both subtrees, excluding their outgroups, form clans in the full tree $\T$. Thus, correct merging amounts to detecting for $\tilde \T_1$ and $\tilde \T_2$ the location of their respective root nodes. To this end, we use an outgroup approach, where the root is set to the node adjacent to the outgroup. 

Let $h_1,h_2$ be the nodes in $\tilde \T_1,\tilde \T_2$ connected to the outgroups ${O_2},{O_1}$, respectively. 
In the merging step we remove the outgroups ${O_2},{O_1}$ and their edges to $h_1$ and $h_2$ from $\tilde \T_1$ and $\tilde \T_2$,
and add an edge connecting $h_1$ to $h_2$. 

\begin{figure}[t]
 \centering \includegraphics[width=1.0\linewidth]{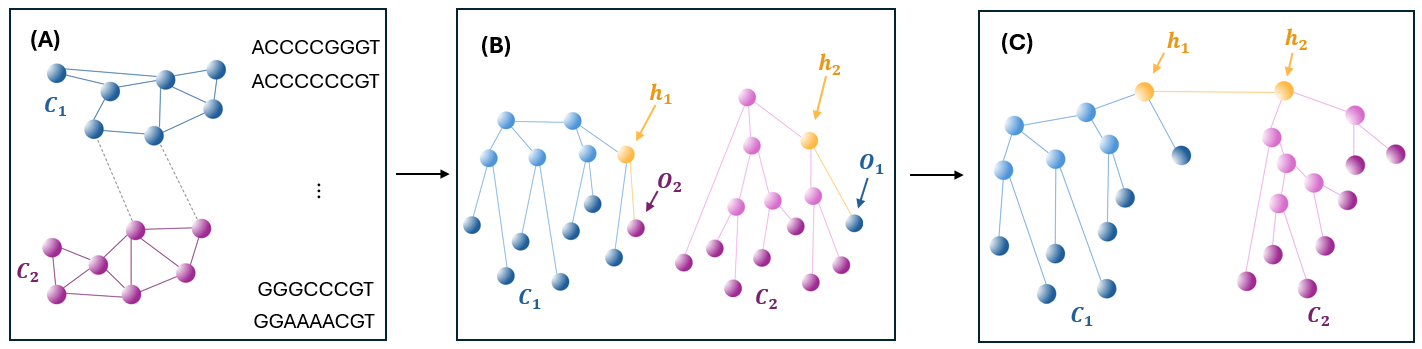}
 \caption{\texttt{SDSR} workflow: (A) Compute a graph based on the genes' similarity matrices, where nodes represent species. The nodes are colored according to the partitions $C_1, C_2$ from step 1. (B) Form subsets $\tilde C_1$ and $\tilde C_2$ by adding outgroups $O_2,O_1$ to $C_1$ and $C_2$, respectively, and reconstruct subtrees $\tilde \T_1$ and $\tilde \T_2$.
 (C) The roots of the two trees are set to $h_1$ and $h_2$, the nodes adjacent to the outgroups $O_2$ and $O_1$. The outgroup nodes are removed, and the subtrees are merged by connecting $h_1$ to $h_2$.}
 \label{fig:merging}
\end{figure}

\begin{remark}
While the recursive structure of \texttt{SDSR} is similar to \texttt{STDR} \cite{aizenbud2021spectral}, we incorporate significant changes to all of the steps in the algorithm.
(i) Unlike \texttt{STDR}, which reconstructs a tree from a single gene alignment, \texttt{SDSR} addresses the challenge of tree reconstruction by merging information from multiple genes. 
(ii) In \texttt{SDSR}, we partition the species by applying constrained $k$-means to the Fiedler vector. This procedure often yields more accurate partitions that those obtained with the sign pattern of the Fiedler vector. (iii) In \texttt{STDR}, the root was determined via a score function computed for all edges of the subtrees. The merging step in \texttt{SDSR} is based on an outgroup approach, which we found to be more accurate and requires fewer computations. 
\end{remark}

\section{Computational complexity of \texttt{SDSR}}\label{sec:complexity}

The total computational complexity of \texttt{SDSR} is composed of three parts: (i) the computation of $K$ similarity matrices; (ii) the recursive partitioning and merging steps; and (iii) the recovery of the small trees via a tree reconstruction subroutine. 
Regarding the first part, computing $K$ similarity matrices from gene alignments of size $m \times n$ requires $\OO(Km^2n)$ operations.

\paragraph{Complexity of recursive partitioning and merging steps.}
Partitioning the $m$ species requires computing the Fiedler vector of a positive semi-definite Laplacian matrix, which has a complexity of $\OO(m^2)$ \cite[Chapter 2]{stewart2001}.
Applying the outgroup approach in the merging step involves identifying the node adjacent to the outgroup in each subtree, removing the outgroup nodes, and adding a single connecting edge between the subtrees. The complexity of these operations is linear in the number of terminal nodes. Thus, the complexity of the partition and merging step \textit{for a single iteration}, given the reconstructed subtrees, is $\OO(m^2)$.
Let $T(m)$ denote the complexity of all partitioning and merging operations.
Given an input of $m$ species, the partition satisfies $\min(|C_1|,|C_2|) \geq \beta m$, for some constant $0<\beta \leq 0.5$. 
Thus, the complexity of recursively reconstructing a tree from $m$ terminal nodes satisfies the following formula,
\begin{equation}\label{eq:recurrence}
T(m) \leq  \max_{\beta \leq \beta' \leq 0.5}T(\beta' m) + T( (1-\beta')m) +  C m^2,   
\end{equation}
for some constant $C$. 
By Lemma \ref{lem:complexity_k2logk} 
in Appendix \ref{sec:complexity_appendix}, 
the complexity $T(m)$ may be bounded by 
\[
T(m) = \OO\Big(m^2 \Big), 
\]
where the multiplicative constant hidden in the $\mathcal{O}\bigl(\cdot\bigr)$ notation depends on $\beta$. 

\paragraph{Complexity of reconstruction of the small trees.}
Given a threshold $\tau$ for the partitioning step, the smallest possible subset size is $\tau \beta$. Thus, the maximal number of subsets is bounded by $m/(\tau \beta)$. For each of these subsets we apply the user-defined tree reconstruction subroutine. Let $f(\tau)$ denote the complexity of reconstructing a tree with $\tau$ terminal nodes via the user-chosen subroutine. 
The complexity of reconstructing all the small trees is thus $\OO(m f(\tau)/\tau)$.
\paragraph{Overall complexity of \texttt{SDSR}.} Combining the complexity of all \texttt{SDSR} steps yields
\begin{equation}
 \OO\Big(Km^2n\Big) +
 \OO\Big(m^2 \Big) + \OO\Big (\frac{m}{\tau} f(\tau)\Big). 
        \label{eq:total_complexity}
\end{equation}
In practice, the actual runtime of the first term in the expression above is negligible, and moreover it can trivially parallelized. 
In addition, computing the Fiedler eigenvector is also in practice much faster than applying the tree reconstruction subroutine. 
Taking these points into account, let us illustrate the implications of the above formula, for a species tree reconstruction whose complexity for $m$ species scales as $f(m)=\OO(m^2)$. In this case, choosing $\tau=\sqrt{m}$ the complexity of the last term is $\OO(m\sqrt{m})$, which represents a savings of $\OO(\sqrt{m})$ as compared to applying the algorithm to the full data. 
In the simulation section we indeed observe runtime savings in accordance to this analysis. 
Next, we present a more detail comparison for \texttt{CA-ML} as a base routine. 

\paragraph{Comparison to \texttt{CA-ML}.} 
We compare the complexity of \texttt{CA-ML} to the complexity of \texttt{SDSR} with \texttt{CA-ML} as subroutine.
\texttt{CA-ML} methods compute a local maximum of the likelihood function of the concatenated matrix by \textit{hill climbing} iterations. The dominant factor in terms of complexity is the Subtree Pruning and Regrafting (SPR) moves, where a subtree is removed and re-attached at a different location. Testing all possible changes based on the concatenated dataset has a $\OO(Km^2n)$ complexity for a single SPR move \cite{allman2005lecture, togkousidis2023adaptive}. This step is applied multiple times in many popular ML tools for tree reconstruction. 
In contrast, using \texttt{SDSR}, 
the cost of reconstructing the small subtrees via ML is only $\OO(K\tau m n)$, a reduction by a factor of $m/\tau$. 
This reduction is achieved at the cost of computing the similarity, which is the dominating factor in \texttt{SDSR} with complexity of $\OO(Km^2n)$. However, this operation is only done once, while SPR moves are performed multiple times in the hill-climbing process.
Thus, the dominant term's complexity is reduced from $\OO(Km^2n)$ to $\OO(K\tau m  n)$. 
As illustrated in the experimental section, for datasets with $200$ taxa, \texttt{SDSR} with \texttt{CA-ML} as a subroutine, is an order of magnitude faster than applying \texttt{CA-ML} to recover the full tree.

\section{Theoretical guarantees under the multispecies coalescent model}\label{sec:analysis}

In this section, we derive theoretical guarantees for the accuracy of the partition step of 
\texttt{SDSR}, 
when the number of genes is large. The guarantees are derived for a partitioning scheme that slightly differs from  
\texttt{SDSR} in two aspects: (i) it uses the unnormalized Laplacian matrices $L^{g} = A^{g}-S^{g}$, instead of the normalized Laplacians defined in Eq. \eqref{eq:laplacians}.
As the un-normalized $L^{g}$ is linear in $S^{g}$,  this allows us to first average the $K$ similarity matrices, and then compute its corresponding Laplacian. 
(ii) The partitioning is determined by the sign pattern of the Fiedler vector, instead of the constrained $k$-means algorithm. 
Our analysis assumes that the un-normalized Laplacian matrices of individual genes are perfectly known. As mentioned in the introduction, this is analogous to the assumption made in the analysis of summary methods, that their gene trees are inferred without errors.

Our analysis assumes the MSC model for generating random gene trees given a species tree and 
the General Time Reversal (GTR) model for generating the sequence alignments given a gene tree. For our paper to be reasonably self-contained we briefly review these two probabilistic
models in Section \ref{sec:data_model}.

Our analysis consists of the following three parts:
\begin{itemize}[leftmargin=*]
\item[\RNum{1}] In Section \ref{subsec:rank_1}, we define a family of matrices $\WT$ that satisfy a \textit{rank-1 condition} with respect to the species tree $\T$. 
We prove that for any graph with a weight matrix $W \in \WT$, partitioning the species according to the sign pattern of its Fiedler vector yields two clans in the species tree.

 \item[\RNum{2}] %In the MSC model, the gene trees are generated by some distribution that depends on the species tree.
 In Section \ref{sec:key_properties} we present two key properties 
 regarding the distribution of the gene trees, 
 used in our analysis. 
 We prove that if these properties hold, then 
 combined with the GTR model, the expected gene similarity matrix satisfies the rank-1 condition as defined in \RNum{1}. 
 Next, we show that the MSC model indeed satisfies these two condition. Hence, 
in the limit $K \to \infty$, partitioning the terminal nodes via \texttt{SDSR} yields two clans of the species tree. 

\item[\RNum{3}] Going beyond asymptotic consistency, in Section \ref{subsec:finite_gene_garantee}, we present a guarantee for correctness of the partitioning step, for a finite number of genes. 
 \end{itemize}

\subsection{Generative models for gene trees and their sequences}
\label{sec:data_model}

In our analysis, we assume that the observed
data was generated according to a combination of the following two probabilistic models. 
Given a species tree, 
we assume that the random 
gene trees are generated by the Multispecies Coalescent (MSC) model.
Next, for each gene tree, 
we assume its corresponding sequences are generated by the General Time Reversal (GTR) Markov model, similar to \cite{allman2019species}.
We briefly describe these two models and then introduce the similarity measure between species, a key quantity in our analysis.

\paragraph{Multispecies Coalescent Model (MSC).} 
The MSC is a common probabilistic model for gene trees given an underlying species tree. 
For a precise description, we refer the reader to \cite[Chapter.~9]{yang2014molecular} and \cite{allman2019species}.
Here, we present only the necessary details required for our analysis. 
%
%In the species tree $\mathcal T$, each edge corresponds to a species, also referred to as \textit{population}. Each internal node of
%$\mathcal T$ represents a divergence event of a population into two distinct populations. 
We start with a definition of the species tree.

\begin{definition}[Species Tree] \label{def:species_tree}
    A species tree is a binary rooted tree, $\mathcal{T}=(V,E)$, where each node $h$ has an associated time $\tau_h\geq 0$, and each edge $e$ has a coalescence rate  $\lambda(e) > 0$. 
    The times of the different nodes satisfy the following two properties: (1) $\tau_v=0$ for all leaves $v$ of the tree; (2) 
    for each parent node $h'$ connected by an edge
    to a child node $h$, 
    $ \tau_h < \tau_{h'}$. 
\end{definition}

By the above definition, the present time at which we observe the current species is set to $\tau=0$. The history of species divergences is represented by the tree.
Specifically, for any two species $i,j$, 
their divergence time is defined as follows. 
%we denote by $\tau_{ij}$ the divergence time of their two lineages in the species tree.  
Let $h_{ij}$ be the internal node of $\mathcal T$
which is their most recent common ancestor (MRCA). 
This node corresponds to the divergence of the lineages $i$ and $j$. 
The divergence time $\tau_{ij}$ of species $i,j$, is given by $\tau_{ij} = \tau_{h_{ij}}$. 
Condition (2) in the definition of the species tree above implies that divergence events in the more distant past have larger time values.
% \begin{remark}\label{rem:time_ultrametric}
%     By Definition \ref{def:species_tree}, the species tree is ultrametric in time units. 
%     That is, the time difference between every leaf and the root is equal.
% \end{remark}

% Generation of random gene trees

Under the MSC, a gene tree $\mathcal T^g$ is randomly generated from the species tree according to the following bottom-up process. 
First, at time $\tau=0$ the gene tree starts with the same set of current species as in the species tree.
Then, as time progresses, the species randomly coalesce to a common ancestor at a rate determined by the species tree topology and  coalescence rates $\lambda (e)$ on its edges.
For any two species $i,j$, we denote by $\tau^g_{ij}$ their coalescence time in the gene tree $\mathcal T^g$. 
The coalescence process satisfies the following two conditions: 

\begin{enumerate}
    \item  $\tau^g_{ij}$ is larger than $\tau_{ij}$. 
    
    \item Their difference, denoted $\Delta^g_{ij}= \tau^g_{ij}- \tau_{ij}$, has an exponential distribution with piecewise constant rates as follows. 
    Let $h_{ij}$ be the MRCA of species $i,j$ in the species tree $\mathcal T$. Let $e_1,e_2,\ldots,e_\ell$ denote the path from $h_{ij}$ to the root of $\mathcal T$. Then, for each branch $e$ connecting a node $h$ to its parent $h'$, conditional on not having coalesced prior to time $\tau_h$, the coalescence rate is $\lambda(e)$. If no coalescence occurred before the root node, then beyond the root, there is a fixed rate $\lambda_0$. 
\end{enumerate}

These two properties imply that the topology of the gene trees may differ from that of the species tree, see illustration in Figure \ref{fig:mscmodel}.

%
%%% I DONT KNOW IF I WANT TO GET INTO THIS RABIT HOLE %%%
% Next, we emphasize  a distinguish between 3 quanteties involved in the MSC model, specificly, three measures for edge length. Let $h$ be a node in the species tree and $h'$ be its parent node. The first and simplest is the time difference $$\tau_{h'}-\tau_h.$$ 
% The second is the coalescence time equal to the time difference multiplied by the coalescence rate, 
% $$\lambda(h,h')(\tau_{h'}-\tau_h),$$
% namely this is a normalization to coalescence units, and from these measures we get the probability for coalescence event of two labeled lineages at that edge by $1-e^{\lambda(h,h')(\tau_{h'}-\tau_h)}$. The last measure is the phylogeny distance, which is the average mutation per site, given by
% $$\mu(h,h')(\tau_{h'}-\tau_h).$$ 
% The last two measures can be generalized by letting the coalescence rate and the mutation rate of each edge be positive functions in time. 
%
Next, we describe the generative process for the observed sequences of the current species in a gene tree.

\paragraph{Substitution model of the sequences.}  
%We now present a probabilistic model of how  sequences are randomly generated, given a gene tree and substitution matrices of the edges of the gene tree. 
%We first describe the Markov model, and then the general time reversal (GTR) model, which we assume in our analysis. 
% GENERAL MARKOV MODEL
We assume that the observed sequences are generated according to the
the common model of sequence evolution, described by a Markov process along the edges of each gene tree (see, e.g.,~\cite[Chapter 1]{yang2014molecular}).
Specifically, we assume that for a given gene tree there are substitution matrices on its edges. 
In addition, at each node $h$ in the tree, there is an associated random variable $X_h\in\{\alpha_1, \ldots,\alpha_d\}$. 
The sequence generative process begins at the root, where its state is drawn according to some initial distribution over $\{\alpha_1,\ldots,\alpha_d\}$.
Then, for any parent node $h'$ connected by an edge to a child node $h$, the random variable $X_h$ is drawn conditionally on $X_{h'}$
via a substitution probability matrix $P_{h|h'}\in \mathbb{R}^{d\times d}$, 
\begin{equation}
P_{h|h'}(l,l') = \Pr(X_{h} = \alpha_{l'} \mid X_{h'} = \alpha_{l}), \qquad \alpha_l, \alpha_{l'} \in \{\alpha_1, \ldots, \alpha_d\}.
    \label{eq:P_h_htag}
\end{equation}
This process propagates from the root to the leaves. 
The observed sequences for each of the current species are generated by this process,  independently for each location in the sequence.

% GTR MODEL
As in \cite{allman2019species}, we assume the GTR model for the substitution matrices. 
Under this model, the sequences mutate along the edges of the tree according to a continuous-time Markov process, with a rate matrix $Q\in \mathbb{R}^{d\times d}$ that satisfies the following four properties: (i) 
$Q_{ij}\geq0$ for all $i\neq j$.
(ii) each row of $Q$ sums to zero.
(iii) there is a distribution vector $\pi\in\mathbb{R}^d$ with positive probabilities that satisfies $\pi Q=0$.
(iv) $\text{diag}(\pi)Q$ is symmetric. 
Under the GTR model, the random variable $X_r$ 
at the root node $r$ of the tree is assumed  to be distributed according to $\pi$.
By the first two properties above, $Q$ is negative semi-definite, namely, all its eigenvalues are non-positive.
If $Q\neq \textbf{0}_{d\times d}$, then there is at least one eigenvalue which is strictly smaller than zero.

Next, we discuss the structure of the substitution matrices under this model.
Let $h,h'$ be two adjacent nodes in the gene tree with times $\tau^g_h,\tau^g_{h'}$. 
Under the GTR model with rate matrix $Q$, $P_{h|h'}, P_{h'|h}$ satisfy:
\begin{equation}\label{eq:GTR_phh'}
    P_{h|h'}=P_{h'|h}=e^{Q|\tau_h-\tau_{h'}|}.
\end{equation}
Using the standard identity for the matrix exponential, 
\begin{equation}\label{eq:det_GTR_phh'}
    \det(P_{h'|h})=e^{\text{tr}(Q)|\tau_h-\tau_{h'}|}.
\end{equation}
Note that $\text{tr}(Q)<0$, hence Eq. \eqref{eq:det_GTR_phh'} implies that $0<\det(P_{h'|h})<1$. 
This condition on the determinants of the substitution matrices is known to be necessary for the identifiability of the tree topology \cite{chang1996full}. 

We note that \cite{allman2019species} considered a slightly more general model, allowing different rate matrices for different genes.
For simplicity, in this work, we consider a fixed matrix $Q$ for all gene trees. 
However, our analysis remains valid under the more general model in \cite{allman2019species}.
%We further discuss this point in Section \ref{sec:discussion}.

\begin{figure}[t]
    \centering
\includegraphics[width=0.6\textwidth,height=0.25\textheight,keepaspectratio]{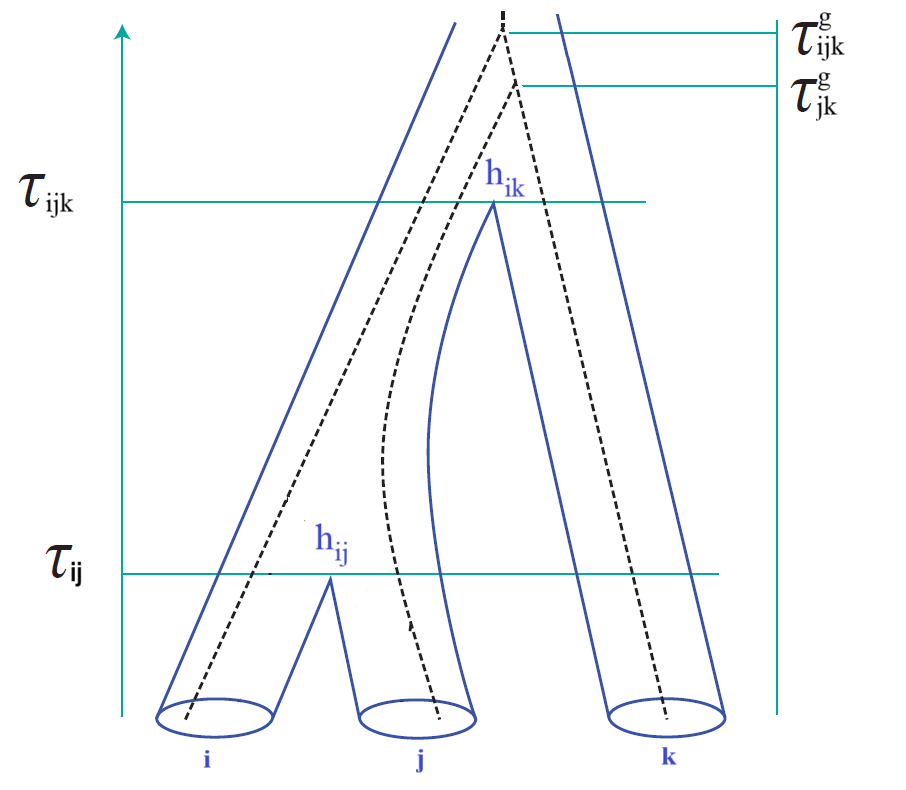}
    \caption{
    {Example of a species tree (blue) and a realization of a gene tree (black). 
    %The species tree  above three species $i,j,k$ alongside its parameters. 
    The vertical line is the time axis starting at $\tau=0$ with the current species. 
    The coalescence times in the species tree are $\tau_{ij}$ and $\tau_{ijk}$, whereas those in the gene tree are $\tau_{jk}^g$ and $\tau_{ijk}^g$.
    In this example, the gene tree and the species tree have different topologies.} }
    \label{fig:mscmodel}
\end{figure}

%%%%%%%%%%%%%%%%%%%%%%%%%%%%%%%%%%%%%%%%

\paragraph{The similarity measure.} 
%Note that under the substitution model described above, the matrix in Eq.~\eqref{eq:sample_substitution_matrix} 
%is an estimate of the underlying substitution matrix between two extant species given by Eq.~\eqref{eq:P_h_htag}.
By the Markov property, the substitution matrix $P^g_{i|j}$ between any two extant species $i$ and $j$ in a gene tree $\mathcal T^g$, is the product of the substitution matrices along the path between $i$ and $j$.
Next, note that the matrix from Eq.~\eqref{eq:similarity} is an estimate of the $m \times m$ \textit{gene similarity matrix},
\begin{equation}\label{eq:population_similarity}
    S^g_{ij} = \det(P^{g}_{i|j}) \cdot \det(P^{g}_{j|i}),
\end{equation}
where $i$ and $j$ are extant species.
Under the GTR model, a closed form expression
for $S^g_{ij}$ can be derived. 
Recall that $\tau^g_i=\tau^g_j=0$, and that $\tau^g_{ij}$ is the time of the MRCA of $i$ and $j$ in the gene tree $\T^g$. Then, 
by Eq.\eqref{eq:GTR_phh'} and the multiplicative property of transition matrices, 
\begin{equation}\label{eq:GTR_similarity}
    S^g_{ij} = e^{4
    \text{tr}(Q)\tau^g_{ij}}.
\end{equation}
Eq. \eqref{eq:GTR_similarity} provides an explicit relation between the coalescence time of species in a gene tree and the similarity measure of their infinite length sequences.
Analogously, we define the \textit{species similarity matrix} $S$ as follows,
\begin{equation}\label{eq:GTR_similarity_species}
    S_{ij} = e^{4\text{tr}(Q)\tau_{ij}},
\end{equation}
where $\tau_{ij}$ is the coalescence time of species $i$ and $j$ in the species tree.
Notice that $-\log(\cdot)$ of the similarity between two species gives the well-known log-determinant distance \cite{steel1994recovering, allman2019species}.

We conclude this section by mentioning the multiplicativity along paths of the similarity measure.
Formally, let $v_1,v_2,v_3$ be three nodes in a tree, such that $v_2$ lies on the unique path from $v_1$ to $v_3$. 
We extend the definition of the similarity in Eq.~\eqref{eq:population_similarity} to a similarity function $\cal S$, which receives as input any pair of nodes $h, h'$ in a tree $\cal T$. It is defined as follows,

\begin{equation}\label{eq:similarity_measure}
    \mathcal S(h,h') = \det(P_{h|h'}) \cdot \det(P_{h'|h}).
\end{equation}
Then, the similarities of the three pairs of nodes satisfy
\begin{equation}\label{eq:multiplicativity}
    \mathcal S(v_1,v_3) = \mathcal S(v_1,v_2) \cdot \mathcal S(v_2,v_3).
\end{equation}
Eq. \eqref{eq:multiplicativity} holds under the general Markov model by the Markov property together with the multiplicativity of determinants. 
Finally, for future use, for any edge $e=(h,h')$ in a tree we define the edge similarity as follows
\begin{equation}\label{eq:edge_similarity}
    \mathcal S(e) = \mathcal S(h,h')
\end{equation}

Since $\text{Tr}(Q)<0$, for each edge $e$ the similarity $0<\mathcal S(e)<1$. 
For future use, we denote the lower and upper bounds of $S(e)$ by $\delta, \xi$ respectively.

\subsection{The species tree rank-1 condition and spectral partitioning}\label{subsec:rank_1}
%Recall that $D$ denotes the exact distance matrix defined in Eq.~\eqref{eq:additivity}, and 
Let $S$ be the similarity matrix of a species tree $\T$ as defined in 
Eq. \eqref{eq:GTR_similarity_species}. The following theorem from \cite{aizenbud2021spectral} establishes the correctness of spectral partitioning 
with the graph Laplacian corresponding to $S$. 
\begin{theorem}[\cite{aizenbud2021spectral},Theorem 4.2]\label{thm:stdr_partition}
Let $\T$ be a binary tree, 
whose similarity matrix $S$ has all entries
between $0$ and $1$. %\footnote{In \cite{aizenbud2021spectral}, the condition is that the similarity is between zero and one. The theory holds, however, for any sim 
Let $v$ be the Fiedler vector of the Laplacian matrix of $S$. Let {$C_1$, $C_2$} be the partition of the terminal nodes of $\T$ according to the sign pattern of $v$. Then $C_1$, $C_2$ are clans in $\T$.
\end{theorem}
Theorem \ref{thm:stdr_partition} is not directly applicable for recovery of the species tree,
since in our setting we do not
have the species tree similarity $S$, but only the similarity matrices $S^g$ of the different
gene trees. 
Toward recovery of the species tree, 
we generalize Theorem \ref{thm:stdr_partition} by introducing a set of matrices, denoted $\WT$, and proving the correctness of spectral partitioning for any weight matrix $W \in \WT$. 
The family $\WT$ consists of all matrices that satisfy the following two conditions. 
\begin{definition}[The $\T$ rank-1 condition]\label{def:rank_one_condition}
  A matrix $W$ satisfies the rank-1 condition with respect to $\T$, if for any partition of the terminal nodes into subsets $A$ and $B$, the submatrix $W(A,B)\in \R_+^{|A| \times |B|}$ is rank-1 if and only if $A$ and $B$ are clans in $\T$.
\end{definition}
For the next condition, let \( W((i,i'),(j,j')) \in [0,1]^{2\times 2} \) denote a submatrix of \( W \) with rows \( (i,i') \) and columns \( (j,j') \). In addition, let \( \T_Q \) denote the subtree of \( \T \) whose terminal nodes are \(Q = ( i,i',j,j') \), with topology induced by \( \T \).

\begin{definition}[The $\T$ quartet condition]\label{def:quartet_condition}
A matrix \( W \in [0,1]^{m \times m} \) satisfies the quartet condition with respect to \( \T \) if, for every quartet of nodes \(Q = ( i,j,i',j') \) such that \( (i,j) \) and \( (i',j') \) are siblings in \( \T_Q \), 
\[
|W((i,i'),(j,j'))| > 0.
\]
\end{definition}

%We denote by $\WT$ the family of matrices in $\R_+^{m \times m}$ that satisfy the above rank-1 condition. 
%The following lemma, which is a variation of  Lemma 3.1 in \cite{snj}, shows that the similarity $S$ is a particular member of $\WT$.
%\begin{lemma}\label{lem:snj_rank1}
% Let $S\in \R_+^{m \times m}$ be the exact similarity matrix of a tree $\T$ with $m$ terminal nodes. 
% Let $A,B$ be a partition of the terminal nodes, and let 
% $S(A,B)\in \R^{|A| \times |B|}$ be the submatrix of $S$ that contains the similarities between nodes in $A$ and nodes in $B$. Then, $S(A,B)$ is rank-1 if and only if $A$ and $B$ contain the terminal nodes of two clans in $\T$. 
%\end{lemma}
The following lemma, proven in Appendix \ref{app:proof_lemma_2}, shows that every matrix $W \in \WT$ corresponds to some tree with the same topology as $\T$. 
\begin{lemma}\label{lem:weight_assignment}
Let $W \in \WT$ be an arbitrary matrix that satisfies the rank-1 condition and the quartet condition with respect to a tree $\T$. 
Then there exists an assignment of edge weights in \( \T \), taking values in \( [0,1] \), such that for all pairs $(i,j)$, $W(i,j)$ is equal to the product of weights along the path from node $i$ to node $j$.
\end{lemma}

Lemma \ref{lem:weight_assignment} implies that though $W$ is not necessarily the similarity matrix of $\T$, there is a different tree $\tilde \T$, with the same topology as $\T$ and different edge similarities for which $W$ is its similarity matrix. We use this result to prove that any matrix $W \in \WT$ can be used for spectral partitioning. 

%\noindent The following theorem shows that the rank-1 property is a sufficient condition
%for both the bottom-up SNJ and the top-down STDR.
\begin{theorem}\label{thm:snj_stdr}
 Let $G$ be a graph whose nodes correspond to the terminal nodes of a tree $\T$, and let $W$ denote its weight matrix. 
 %The weight between two nodes of the graph $x_i,x_j$ is equal to the similarity measure $S(i,j)$ as defined in Eq. \eqref{eq:similarity}. 
 %Let $S \in \R_+^{n \times n}$ be the similarity matrix (weight matrix) of the graph $G$ and 
 Let $v$ be the Fiedler vector of the graph's unnormalized Laplacian matrix $L$. If $W \in \WT$, then partitioning the terminal nodes according to the sign pattern of $v$ yields two clans in $\T$.
\end{theorem}
%\noindent Put simply, the rank-1 condition is sufficient for both the SNJ STDR.
\begin{proof}
%The proof for theorem $\ref{thm:snj_stdr}$ is based on two components: (i) 
By lemma \ref{lem:weight_assignment}, if a matrix $W$ satisfies the rank-1 condition, then there is a tree $\tilde \T$ with the same topology as $\T$ such that elements $W(i,j)$ are equal to the product of edge weights of $\tilde \T$, along the path from node $i$ to node $j$. 
Theorem \ref{thm:stdr_partition} (Theorem 4.2 from \cite{aizenbud2021spectral}) proves that partitioning the terminal nodes according to the Fiedler vector of a graph with weights $W$ yields two clans of $\tilde \T$. Since the topology of $\tilde \T$ is identical to that of $\T$, the partition corresponds to two clans of $\T$. 
\end{proof}
%As mentioned, Theorem \ref{thm:snj_stdr} generalizes over Theorem 4.2 from \cite{aizenbud2021spectral}, which has a similar guarantee specifically for the similarity matrix $S$. 

%Lemma 3.1 in \cite{snj} proves that under the assumption in Eq. \eqref{eq:delta_xi}, the similarity matrix $S$ satisfies the rank-1 condition $S \in \WT$. Thus, Theorem \ref{thm:snj_stdr} that proves correctness for the particular case where the graph's weight matrix equals $S$ can be derived from Theorem \ref{thm:stdr_partition}.
 
 \subsection{Asymptotic consistency of the partition step}
\label{sec:key_properties}

We now prove the consistency of the partitioning step of \texttt{SDSR}. 
The proof relies on Theorem  
\ref{thm:snj_stdr} combined with the following lemma,
which shows that under the MSC model, the expected similarity $\E[S^g]$ satisfies the rank-1 condition. 
The proof of the lemma is in Appendix \ref{sec:proof_theorem_partition}.
\begin{lemma}\label{lem:population_mean}
Let $\T$ be a species tree, and let $\T^{g}$ be a gene tree generated according to the MSC model. 
Let $\E[S^{g}]$ be the expected similarity matrix over the distribution of $\T^{g}$. 
Then $\E[S^{g}]$ satisfies the rank-1 condition and the quartet condition with respect to $\T$. %Partitioning the terminal nodes according to the sign pattern of the Fiedler vector of $L$ yields two clans in the species tree.
\end{lemma}
%For a given tree $\mathcal T$ and a subset $A \subset [m]$ of terminal nodes, we denote by $h_{A}$ the (non-terminal) node that is the most recent common ancestor (MRCA) of all nodes in $A$. 
%When $A$ consists of a pair of nodes $\{i,j\}$, we denote their MRCA by $h_{i,j}$. 
%We denote by $D(i,h_{i,j})$ the distance in a gene tree between terminal node $i$ and non-terminal node $h_{i,j}$. 
%Under the coalescent model, 
%as illustrated in Figure \ref{fig:msc_model}, 
%species $i,j$ may coalesce higher up the tree.
%In the coalescent model, the difference between the gene tree and species tree distances, $D(i,j)-D(i,j)$, is due to a random shift in the coalescence time. 
%That is, there is a non-negative random variable $\Delta_{ij}$ such that
%\begin{equation}\label{eq:Delta_ij_def}
%    \Delta_{i,j} = D(i,h_{i,j})-D(i,h_{i,j}) = D(j,h_{i,j})-D(j,h_{i,j}),    
%\end{equation}
%where $h_{i,j}$ is the MRCA of $i,j$
%in the $k$-th gene tree, which may differ from 
%their MRCA $h_{i,j}$ in the species tree. 
%By the additive property of distance measures from Eq. \eqref{eq:additivity},
%\begin{equation}\label{eq:D_g}
%D(i,j) = D(i,j) + 2\Delta_{i,j}.
%\end{equation}
Recall that $\Delta_{i,j}^g$ denotes the difference between the coalescence time of gene $g$ in the gene tree and the corresponding coalescence time in the species tree. 
For our proof of consistency we use the
following two properties. 

%The second condition of the MSC model implies that the distribution of $\Delta_{i,j}^g$ is fully determined by the edge parameters $\lambda(e)$ along the path from $h_{i,j}$ to the root of the species tree. In particular, this condition implies the following two properties.
%the following two properties on the distribution of $\Delta_{i,j}^g$
%We assume that the distribution of $\Delta^{g}}_{i,j}$ satisfies the following two key properties, as illustrated in Figure \ref{fig:quartet}.
%\begin{enumerate}
\begin{prop}\label{prop_1}
 %[(i)] 
 For any two pairs of terminal nodes, $(i,j)$ and $(i',j')$, with the same MRCA in the species tree,  $h_{i,j} = h_{i',j'}$, the random variables $\Delta_{i,j}^{g}$ and $\Delta_{i',j'}^{g}$
 have the same distribution.
 \end{prop}
%if $h_{i,j} = h_{i',j'}$, the distribution of $\Delta^g$
 \begin{prop}\label{prop_2}
 %\item[(ii)] 
 For any triplet of terminal nodes $\{i,j,l\}$, where $h_{i,j}$ is the MRCA of $i,j$ and $h_{i,l} = h_{j,l}$ is the MRCA of all three nodes, $\Delta_{i,j}^g$ is stochastically dominated by $\tau_{i,l}-\tau_{i,j} + \Delta_{i,l}^{g}$, such that
 \[
 \Pr\big(\tau_{i,l}-\tau_{i,j} + \Delta_{i,l}^{g} \geq t \big) > \Pr(\Delta_{i,j}^{g} \geq t) \: , \quad \forall t > 0.
 \]
 \end{prop}
%\end{enumerate}

\begin{remark}
Note that under the MSC model, both properties 1 and 2 above indeed hold.
Indeed, property 1 holds since
the distribution of $\Delta^g_{i,j}$ and $\Delta^g_{i',j'}$ is determined by the edge parameters $\lambda(e)$ on the edges between the MRCA and the root, which are the same of both pairs.
As for the second property, 
let $A$ be the event that the pair $i,j$ coalesce before $\tau_{i,l}$. In that case $\Delta_{i,j}^g< \tau_{i,l}-\tau_{i,j}$. Under the condition that $A$ did not occur, the MSC implies that the distribution of $\tau_{i,l}-\tau_{i,j} + \Delta_{i,l}^{g} $ is identical to that of $\Delta^g_{i,j}$. Since the probability of $A$ is positive, we conclude that the second property holds.  
\end{remark}

%Property 2 holds since there is a probability that the coalescent event of $i,j$ will occur between $\tau_{i,j}$ and $\tau_{i,j}$, such that   $\Delta_{i,j}^g< \tau_{i,l}-\tau_{i,j}$. If however, $i,j$ did not coalesce before $\tau_{i,l}$ then  
%coalescen gene coalesence time $\tau^g_{i,j}$ will be smaller than $\tau_{i,l}$
The following theorem states the consistency of the partitioning step under properties 1 and 2. 

\begin{theorem}\label{thm:partitioning_consistency}
Let $\{\T^{g}\}_{g=1}^K$ be $K$ gene trees, generated independently from a species tree according to a distribution that satisfies properties 1 and 2 above. Let $\bar{S}= \frac1K \sum_{g=1}^K S^{g}$ be the empirical mean of their (exact) similarity matrices $S^{g}$
given in \eqref{eq:GTR_similarity}. 

Assume that all entries of the Fiedler vector of the Laplacian matrix corresponding to $\mathbb{E}[S^{g}]$ are non-zero, and that the smallest non-zero eigenvalue has multiplicity 1. Then, as $K\to\infty$, 
the partition computed by \texttt{SDSR} yields two clans in the species tree, with probability tending to one.     
\end{theorem}
Theorem \ref{thm:partitioning_consistency} shows that the partitioning scheme yields two clans 
in the limit as the number of genes $K\to\infty$. 
For simplicity, the theorem is stated for the case of a single (top) partition. It can, however, be applied recursively to prove consistency in multiple partitions.

\subsection{Guarantee for correct partition with a finite number of genes}\label{subsec:finite_gene_garantee}

In the previous section, we discussed the asymptotic consistency of the partition step, 
as the number of genes tends to infinity. 
In this section, we present in Theorem \ref{thm:finite_gene_guarantee_v2} a guarantee for correct partitioning, which holds with high probability for a finite number of genes.
For this analysis, we make the assumption that the species tree is ultrametric, that is, all its terminal nodes are at the same distance from the root. 
This assumption holds, for example, under the MSC+GTR model. 
It implies that the smallest non-zero eigenvalue of the relevant graph Laplacian matrix has multiplicity one, see  \cite[Lemma 7]{balakrishnan2011noise}.

To present this result, we first introduce some definitions. 
We start by defining $\eta$, a measure of the \textit{imbalance} at the root of the species tree $\cal T$. Specifically, let $h_R$
and $h_L$ be the two children of the root $r$, and let $C_L$
and $C_R$ be the observed species that are descendants of $h_L$ and $h_R$ respectively.
The imbalance parameter is 
$\eta = \frac{\max\{|C_L|,|C_R| \} }{ \min\{ |C_L|,|C_R| \} }$. 
%
%For any internal node $h$, let $C_L(h),C_R(h)$ be 
%number of observed species that are descendants of its left and right child in the species tree $\mathcal T$, respectively. The imbalance parameter is defined as follows,
%    \(
%    \eta = \max_h \left\{ \frac{C_L(h)}{C_R(h)}, \frac{C_R(h)}{C_L(h)} \right\}.
 %   \)
Next, we define $\mathcal{T}_E$ to be the tree with the same topology as $\mathcal{T}$, whose similarity matrix is $\mathbb{E}[S^{g}]$. The existence of this tree follows from 
Lemma~\ref{lem:weight_assignment}
combined with the fact that $\mathbb{E}[S^{g}] \in \mathcal{W_T}$,
as stated in Lemma \ref{lem:population_mean}.
We denote $S_{\min}=\min_{i,j\in [m]}\E [S^{g}(i,j)]$, where the minimum is over all pairs $i,j$ of terminal nodes.
%In analogy to the definition of $\delta$ in \eqref{eq:delta_xi}, we define $\delta_E$ as a lower bound on the similarities between adjacent nodes in $\mathcal{T}_E$, namely $S_E(i,j)\geq \delta_E$ for all adjacent nodes $i,j\in\mathcal{T}_E$. 
Next, we introduce a 
\textit{separation parameter} $s_r$, which quantifies the separation between the two clans of the first partition of $\mathcal{T}_E$. 
As above, let $h_L,h_R$ be the two children of the root. 
Then, $s_r=\max\{S_E(r,h_L), S_E(r,h_R)\}$, 
where $S_E(h,h')$ is the similarity
between two nodes $h,h'\in\mathcal{T}_E$. 
%Recall the definition of $\delta$ from subsection \ref{subsec:rank_1}. Analogously, let $\delta_E$ be the constant satisfies $0<\delta_E\leq S_E(i,j)$ for all pairs $i,j$ of adjacent nodes in $\mathcal T_E$. 
We are now ready to present our high-probability guarantee for correctness of the partitioning step.

\begin{theorem}\label{thm:finite_gene_guarantee_v2}
    Let $\mathcal T$ be an ultrametric species tree with $m$ terminal nodes and imbalance parameter $\eta$. 
    Let $h_L,h_R$ be the two children of the root $r$ of $\mathcal T$ and let $C_L, C_R$ be their observed descendant species, respectively.
    Let ${\mathcal{T}^{g}}_{g=1}^K$ be $K$ i.i.d. gene trees generated according to the MSC+GTR model,  
    and let $\{S^{g}\}_{g=1}^K$ be their exact similarity matrices. 
    Assume that for some $\epsilon>0$, the number of genes $K$ satisfies 
\begin{equation}\label{eq:RHS_bound}
    %m^{C_3 \log(1/\delta_E)}
    K\geq
    \begin{cases}
    c_1 \cdot \frac{\eta^{3}}{(1 - s_r^2)^2S^2_{\min}} \cdot \log(\tfrac{m}{\epsilon}) & \text{if } \xi \leq \tfrac{1}{2}, \\
    c_2 \cdot \frac{\eta^{3}}{(1 - s_r^2)^2S^2_{\min}} \cdot m^{2 + 2\log_2 \xi}\cdot\log(\tfrac{m}{\epsilon}) & \text{if } \xi > \tfrac{1}{2}.
    \end{cases}
    \end{equation}
    where $c_1, c_2$ are universal constants. 
    Then, with probability larger than $1-\epsilon$, partitioning based on the Fiedler vector corresponding to $\bar S=\frac1K \sum_g S^{g}$ yields the two clans $C_L$ and $C_R$ of the species tree $\mathcal T$.
\end{theorem}

% Discussion
Let us make several remarks regarding this theorem. 

%\textcolor{red}{Ofer: I think the next remark can be removed}
\begin{remark}
    Theorem \ref{thm:partitioning_consistency}
    guarantees a correct partition into unspecified two clans. In contrast, 
    Theorem \ref{thm:finite_gene_guarantee_v2}
    specifies that the partition is to the two clans induced by removal of the root of the tree. This more specific result is due to the assumption that the species tree is ultrametric.       
\end{remark}
\begin{remark}
    As expected, as the separation parameter $s_r \to 1$, namely as the clans become poorly separated, the number of genes required to guarantee a correct partition increases to infinity. 
\end{remark}

Theorem \ref{thm:finite_gene_guarantee_v2} provides a guarantee for the partitioning step of \texttt{SDSR}. 
Given that the partition is accurate, the precision of the reconstruction and merging steps depend on the outcome of the user-defined subroutine.
The following proposition extends Theorem \ref{thm:finite_gene_guarantee_v2} by providing a theoretical guarantee for the full 
\texttt{SDSR} pipeline.  
\begin{proposition}\label{prop:sdsr}
Assume that the assumptions of Theorem \ref{thm:finite_gene_guarantee_v2} hold, and that the number of gene alignments $K$ satisfies Eq. \eqref{eq:RHS_bound}. If, for both subsets $\tilde C_1$ and $\tilde C_2$ the subroutine computes a tree topology that matches the subtree of 
$\T$ induced by the respective subset, then 
\texttt{SDSR} returns a tree with the same topology as $\T$, with probability $1-\epsilon$.
\end{proposition}
Similarly to the partitioning guarantees, Proposition \ref{prop:sdsr} is written for the case of a single partition,  but can be applied recursively to show a similar result for further partitions. 
%Instead, assume a noisy similarity $\hat{S}_g$ whose level of noise is determined by the length of the sequences, $n$, from which $\hat{S}_g$ is estimated. 

\section{Simulations}\label{sec:experiments}

We present an empirical evaluation of \texttt{SDSR} on multiple simulated datasets with ILS and HGT events. 
The datasets and the rates of their HGT and ILS events are described in Section \ref{sec:data_description}. 
As the success of \texttt{SDSR} crucially depends on accurate partitioning of the species into clans, 
in Section \ref{sec:partition} we evaluate the accuracy of this partitioning step. 
In Section \ref{sec:testing_recunstruction_and_runtime}
we evaluate the reconstruction accuracy and runtime of \texttt{SDSR}.
For comparison, we use the following methods:
(i) the \texttt{RAxML} maximum likelihood package \cite{stamatakis2014raxml} applied to a concatenation of the gene alignment matrices (denoted as \texttt{CA-ML}), and (ii) \texttt{ASTRAL-IV} \cite{zhang2018astral, zhang2022weighting}.
We compare the outcome of \texttt{CA-ML} and \texttt{ASTRAL-IV} to the results obtained via \texttt{SDSR} with these methods as subroutines for recovering small trees. 
Finally, on the 10,000 species repository we compare our results to \texttt{uDance} \cite{balaban2024generation}, a state-of-the-art divide-and-conquer approach for species tree reconstruction \footnote{A python implementation of SDSR is available at \url{https://github.com/reshefo/sdsr} }. 

\subsection{Datasets }\label{sec:data_description}

We considered the following three publicly available repositories. Each repository includes several datasets with a simulated species tree, a set of corresponding gene trees, and their gene alignments. These repositories have been used to evaluate species tree reconstruction methods in several previous studies.  

\textbf{$50$-species datasets.} This repository 
consists of $100$ 
datasets from \cite{tandy1}, each containing $1,000$ gene alignments of $50$ species, with sequence length of $997$bp for each gene. 
The first $50$ datasets include only ILS events, with a mean nRF distance of $0.3$ between the true gene trees and the true species tree. The other $50$ datasets include, in addition to ILS events,  HGT events at an average rate of $20$ events per gene.

\textbf{$200$-species datasets.} 
This repository consists of $50$ datasets from \cite{zhang2023caster}, each containing a simulated concatenated genome alignment of $200$ species, where the simulations incorporate ILS events. As in \cite{zhang2023caster}, the sequence is divided into genes of size $2500$bp.

\textbf{$10{,}000$-species datasets.} 
This repository consists of $10$ datasets from \cite{balaban2024generation}, each containing $500$ gene alignments. The gene sequences ranges in length from $169$ to $869$ base pairs, with an average of $406$ base pairs.
In this dataset, ILS and HGT rates were set such that the mean nRF distance between the species tree and gene trees was $0.03$ due to ILS only and $0.44$ when adding the HGT events.

\subsection{Partition Accuracy}
\label{sec:partition}

First, we evaluate the accuracy of the partitioning step of \texttt{SDSR} as a function of the number of gene alignments $K$ for the $10{,}000$-species datasets.
The imbalance parameter $\beta$ was set to $0.01$ and the threshold parameter $\tau$ to $500$ species.
We compare the estimated partition to the most similar clan in the ground-truth species tree.
Figure \ref{fig:10000_Species_Datasets_Partition_based_similarity} shows the Rand Index score as a function of the number of genes. The left panel shows the score for the first partition only, and the right shows the average score for all partitions. 
The points mark the median of the $10$ datasets, and the line is the range between the $25\%$ and $75\%$ quantiles.
The results show that as the number of genes increases, the estimated partitions approach the ground-truth partitions. A similar experiment was done on the $50$-species and $200$-species repositories. The results, presented in Appendix \ref{sec:200_50_partitions}, show that $K=50$ alignments were sufficient for accurate partitioning for both datasets. In addition, we compared the partitioning accuracy to a variation of our approach where the average is taken over the distance matrices instead of the similarity matrices. The results, presented in Appendix \ref{sec:distance_partition}, show that this approach yields partitions with considerably more errors than \texttt{SDSR}.

\begin{figure}[t]
\centering
\makebox[\textwidth][c]{%
 \includegraphics[width=0.7\textwidth]{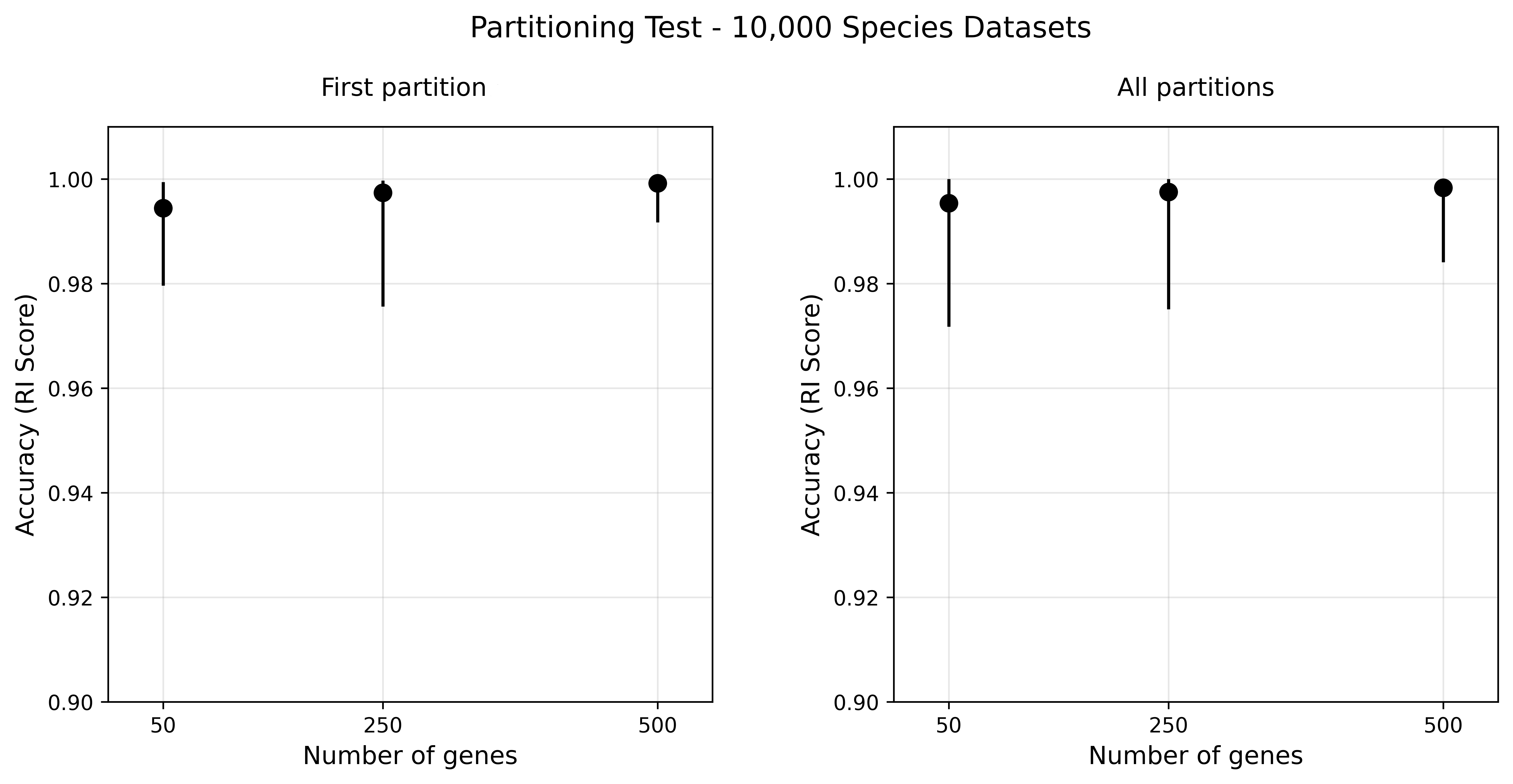}%
}
\caption{Partition accuracy as a function of the number of genes, on the $10{,}000$-species datasets. Results shown for the first partition (left) and all partitions (right). The accuracy of the first partition is measured by the Rand Index similarity score 
to the most similar partition in the ground-truth species tree. 
For internal partitions, we compared the obtained partition to the most similar partition in the corresponding subtree.}
\label{fig:10000_Species_Datasets_Partition_based_similarity}
\end{figure}

\begin{figure}[h]
 \centering
\includegraphics[width=0.75\textwidth]{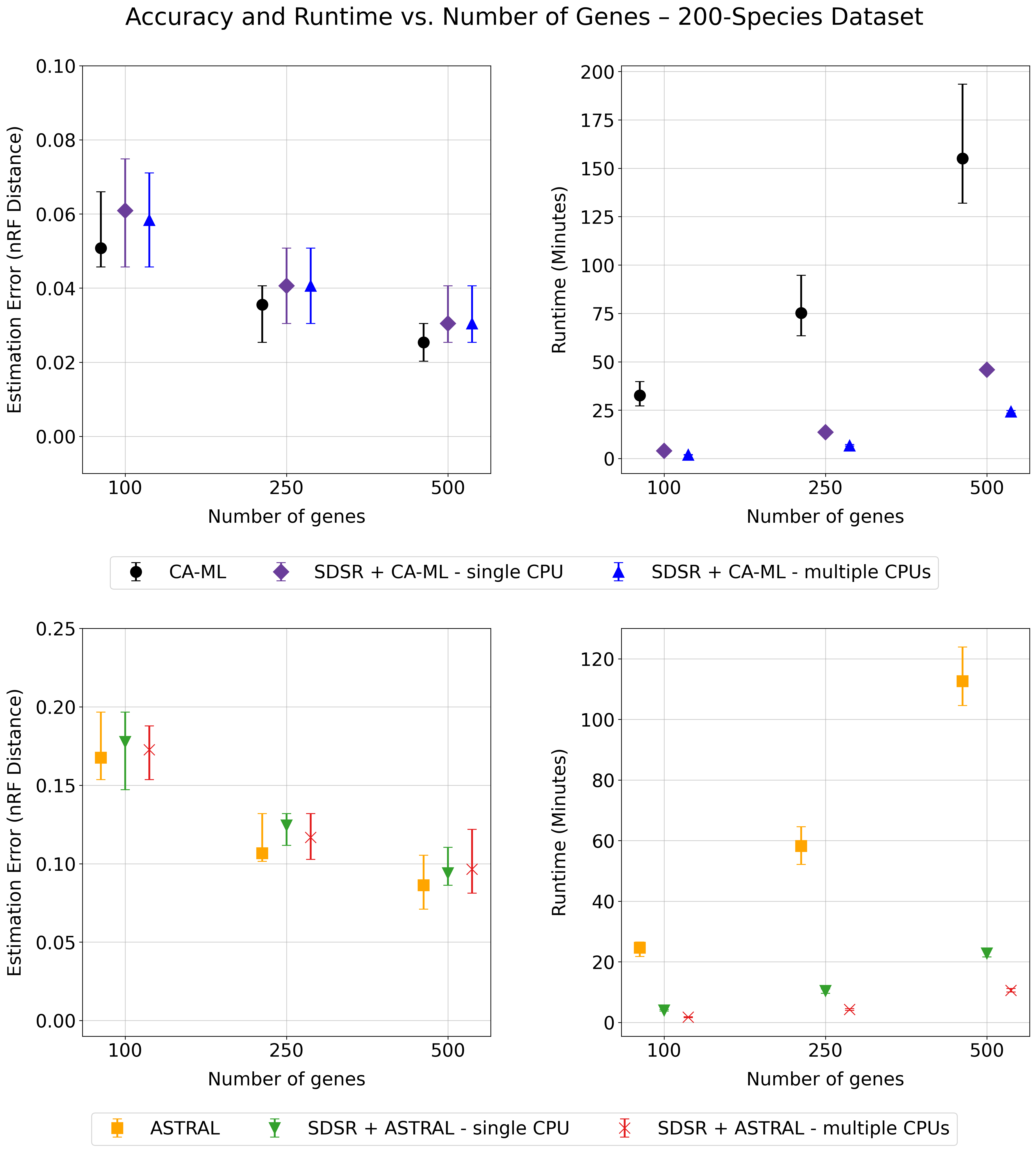}
 \caption{
 Accuracy and runtime as a function of the number of genes for the $200$-species datasets. The upper panel presents the nRF distance between the estimated and ground truth species tree of \texttt{CA-ML} and \texttt{SDSR} using \texttt{CA-ML} as a subroutine. The upper right panel presents the runtime. The bottom panels present the accuracy and runtime of \texttt{ASTRAL} and \texttt{SDSR} using \texttt{ASTRAL} as a subroutine.
 }
\label{fig:200_Species_Datasets_caml_verses_sdsr}
\end{figure}

\subsection{Reconstruction accuracy and runtime}\label{sec:testing_recunstruction_and_runtime}

We compared the accuracy and runtime of \texttt{ASTRAL} and \texttt{CA-ML} when used as stand-alone methods and when used as subroutines of \texttt{SDSR}, as a function of the number of genes.
For \texttt{SDSR}, we set the imbalance parameter $\beta=0.05$ and threshold to $\tau=50$ species. The partition and reconstruction steps were applied to $100$, $250$, and $500$ genes. 
Figure \ref{fig:200_Species_Datasets_caml_verses_sdsr}
shows the results for the 200-species repository.
The figure shows the 
median accuracy and the range between the $25\%$ and $75\%$ quantiles of the $50$ datasets in the repository. The upper row compares (i) \texttt{CA-ML}, (ii) \texttt{SDSR} + \texttt{CA-ML} on a single CPU, and (iii) \texttt{SDSR} + \texttt{CA-ML} on multiple CPUs. The bottom row shows a similar comparison with \texttt{ASTRAL}.
Following the original study from which the $200$-species datasets were obtained \cite{zhang2023caster}, the input to \texttt{ASTRAL} was slightly different than the input to \texttt{CA-ML}.
For \texttt{ASTRAL}, each gene alignment consisted of $500$bp selected out of a $2,500$bp window. For \texttt{CA-ML}, the input consisted of the entire sequences.

As seen in the figure, \texttt{SDSR} combined with either \texttt{CA-ML} or with \texttt{ASTRAL}, achieves similar accuracy as the respective method run as a stand-alone method to recover the full tree, while significantly reducing runtime. 
For example, with $100$ genes, \texttt{CA-ML} and \texttt{SDSR} + \texttt{CA-ML} achieved similar accuracy. We note that the differences in accuracy between applying \texttt{SDSR} on a single and multiple CPUs are due to randomness in the RAxML method.

The runtime of \texttt{SDSR} + \texttt{CA-ML} was $8$ times faster than \texttt{CA-ML} alone on a single CPU, and $17$ times faster when all subtrees are recovered in parallel on separate  CPUs.
The runtime results are consistent with the analysis in Section \ref{sec:complexity}. For a single CPU, the theoretical speedup of \texttt{SDSR} when partitioning the species to equally sized subsets is $m/\tau = 200/50 = 4$. In practice, the partition is to subsets that may be smaller than $50$, which increases the speedup. For multiple CPUs, the reconstruction of the different trees is parallelized. 
Hence, the theoretical speedup is the ratio between the complexity of recovering the full tree and the largest subtree, which is $m^2/\tau^2 = 16$. This is similar to the observed speedup of $17$. For $500$ genes, the runtime speedup is slightly reduced to $11$, due to the additional time required for computing the similarity matrices.

In the Appendix section, we show the results  of additional experiments that evaluate the effect of the threshold $\tau$.  First, we tested the accuracy of \texttt{SDSR} with a single partition. The results, presented in Appendix \ref{sec:partition_experiment_single_iteretion}, demonstrate that \texttt{ASTRAL} or \texttt{CA-ML} combined with \texttt{SDSR} achieves similar accuracy to the original method with reduced runtime. Next, we evaluate the effect of increasing the number of partitions by lowering the threshold parameter $\tau$. The results, presented in Appendix \ref{sec:multiple_partitions}, show that using a lower threshold substantially reduces runtime while maintaining similar accuracy. 

Next, we evaluate \texttt{SDSR} on larger datasets with $10,000$ species. We compare our approach to \texttt{uDance} \cite{balaban2024generation}, a state-of-the-art divide-and-conquer approach developed for large datasets. Both methods were applied in parallel on $32$ CPUs. Additionally, both methods used \texttt{ASTRAL} to reconstruct the subtrees. 
We set the imbalance parameter of SDSR to $\beta = 0.001$ and the partitioning threshold to $\tau = 1,000$. 
We note that for \texttt{SDSR}, partitioning was performed using all $500$ genes across all runs, while reconstruction was applied with $4$, $16$, $64$ and $256$ genes. We also note that \texttt{uDance} excludes some of the species from the output tree. The average number of species in the tree computed by \texttt{uDance} was around 6,500 species for $K=4$ genes and around 9,900 for $K=256$. The results show that for settings where the reconstruction is done with less than $K=64$ genes, \texttt{SDSR} computes more accurate trees
than \texttt{uDance}. However, \texttt{uDance} is more accurate for cases with larger number of genes. 

\begin{figure}[t]
 \centering
 \includegraphics[width=0.95\textwidth]{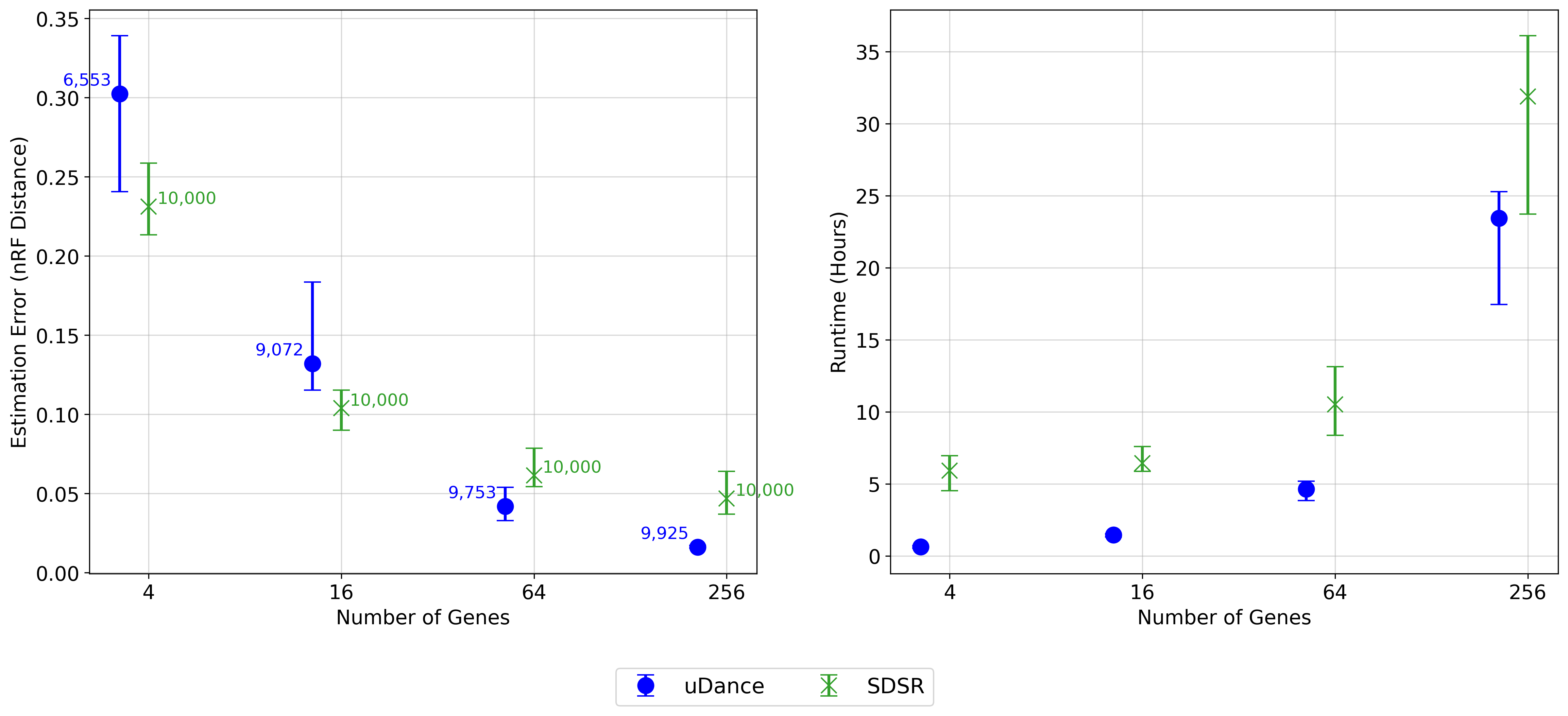}
 \caption{Comparison to \texttt{uDance} \cite{balaban2024generation} on the $10{,}000$-species datasets. The left panel shows the nRF distance between the estimated and ground-truth species tree as a function of the number of genes. Each marker represents the median runtime and nRF distance. The label on each marker denotes the median number of species in the reconstructed trees. 
 The right panel shows the runtime of the two methods.
 }
\label{fig:10000_species_more_data_same_time}
\end{figure}

\section{Discussion}
    \label{sec:discussion}

This paper introduces \texttt{SDSR}, a  {discordance-aware method} for species tree reconstruction. We derived theoretical guarantees under the MSC model for gene trees, and the GTR model for sequence alignments. {\texttt{SDSR} can boost the scalability of many base species tree reconstruction methods. The algorithm reduces runtime even when run on a single CPU and can also be parallelized to leverage multiple CPUs.}
As shown in the simulations above, the time reduction achieved by our method becomes more substantial as the number of species increases. 

There are several directions for future research. One potential approach to improve our method is by 
using a weighted average of the gene Laplacian matrices in the partitioning step, rather than assigning equal weight to each gene. Ideally, the weight assigned to each gene would reflect its reliability, specifically, how closely its gene tree matches the species tree. Since the discrepancy between the gene trees and the species tree is generally unknown, we plan to adapt unsupervised ensemble learning methods \cite{jaffe2015estimating, jaffe2016unsupervised} that estimate reliability based on correlations among gene trees.

Another potential source of error in our approach may arise in the merging step. In the current implementation of the \texttt{SDSR} algorithm, a single outgroup species is used to detect the root of each subtree, after which an edge between the estimated roots connects the two subtrees. However, relying on a single outgroup may be suboptimal, as it ignores potentially informative phylogenetic signals present in additional outgroups. For finding the root of the unrooted tree, it has
become common practice to use more than a single outgroup when such are available \cite{barriel1998rooting, kinene2016rooting}. In future research, we aim to add multiple outgroups to each subtree. Then, develop a machine-learning-based method for identifying the most accurate merging location among those suggested by different outgroups.

In terms of theoretical analysis, we can improve our guarantees by taking into account estimation errors in the similarity matrix. These directions can significantly increase the accuracy of our divide-and-conquer approach, enabling accurate recovery of larger trees without the need for extensive computational resources.

\printbibliography

\newpage

\appendix
\section{Complexity of partitioning and merging steps}
    \label{sec:complexity_appendix}
 
The complexity of the partitioning and merging steps is given in its recurrence form in Sec. \ref{sec:complexity}. The following lemma proves an explicit bound for this complexity.

\begin{lemma}\label{lem:complexity_k2logk}
     \textcolor{black}{Let $f(k)$ be a function that satisfies the recurrence equation 
    \[
    f(k) \leq \max_{ \alpha_0\leq \alpha<0.5} f(\alpha k) + f((1-\alpha)k) + c k^2 
    \]
    for some $0<\alpha_0\leq 0.5$ independent of $k$. Then $f(k) = \OO(k^2)$.}
\end{lemma}
\begin{proof}
     We prove the bound by induction. The induction base is trivial since for small $k$, there is always a constant $C_c$ such that $f(k) \leq C_c k^2 $. Assume for all $k <k_0$ we have that $f(k) \leq C_c k^2$. Then,
    \begin{align*}
    f(k_0) &\leq  \max_{ \alpha_0\leq \alpha<0.5} f(\alpha k_0) + f((1-\alpha)k_0) + c k_0^2 \\
    &\leq \max_{ \alpha_0\leq \alpha<0.5} C_c \alpha^2 k_0^2 + C_c (1-\alpha)^2 k_0^2 + c k_0^2 \\
    &\leq k_0^2\max_{ \alpha_0\leq \alpha<0.5} (C_c \alpha^2 + C_c (1-\alpha)^2+c) .
    \end{align*}
    For large enough $C_c$ we have that $f(k_0) \leq C_c k_0^2$, which concludes the proof.
\end{proof}
\begin{remark}
     In case the tree is split at each level with ratio constant, the computational complexity of the partitioning and merging steps can be derived directly from the Akra–Bazzi method \cite{akra1998solution}.
\end{remark}

\section{Proof of Lemma \ref{lem:weight_assignment}}
\label{app:proof_lemma_2}

To prove the lemma, we first state the following 
auxiliary result on the rank-1 condition. 

\begin{lemma}\label{lem:rank_1_submatrices}
Let $\mathcal T$ be a tree decomposed into
$\mathcal T = \mathcal T_A \cup \mathcal T_B$, where $A$ and $B$ are clans of $\mathcal T$. 
Suppose that $W$ satisfies the rank-1 condition with respect to $\mathcal T$. Then, 
$W(A,A)$ and $W(B,B)$ satisfy the rank-1 condition with respect to $\mathcal T_A$  and $\mathcal T_B$, respectively.
\end{lemma}

\begin{proof}[Proof of Lemma \ref{lem:rank_1_submatrices}]
Let $A=A_1 \cup A_2$ be a partition of $A$
into two clans with respect to the tree $\mathcal T_A$. Our aim is to show that
$W(A_1,A_2)$ is rank-1. To this end, 
note that $A_1$ is also a clan in the original tree $\mathcal T$. Therefore, by the rank-1 condition, $W(A_1,A_2 \cup B)$ is rank-1. This implies that its submatrix $W(A_1,A_2)$ is rank-1 as well. 
\end{proof}

\begin{proof}[Proof of Lemma \ref{lem:weight_assignment}]
We prove the lemma by induction on the number of terminal nodes $m$. 
Our base case is a tree with two nodes $x_1,x_2$ where there is a single edge and a single non-diagonal element $W(1,2)$. Hence, one can simply assign the weight $W(1,2)$ 
to the edge $e(1,2)$ and thus satisfy the lemma.

Next, we assume that the lemma holds for any tree with up to $m-1$ terminal nodes, and prove the claim for a tree $\T$ with $m$ nodes. 
We partition $\T$ into two rooted subtrees $\T_A$ and $\T_B$ by removing a single edge. Let $h_A$ and $h_B$ be the two root nodes of $\T_A$ and $\T_B$, respectively, and let $A$ and $B$ be their corresponding subsets of terminal nodes, see illustration in Figure \ref{fig:similarity_structure}.
Accordingly, we partition the similarity matrix $W$ into four blocks: (i) The submatrix $W(A,A)$ which contains the pairwise similarity between nodes in $A$, (ii) The submatrix $W(B,B)$ which contains the pairwise similarity between nodes in $B$, and (iii) the matrices $W(A,B)$ and $W(B,A) = W(A,B)^T$ that contain similarities between nodes in $A$ and those in $B$.

\begin{figure}[t]
 \centering
 \includegraphics[width=0.8\linewidth]{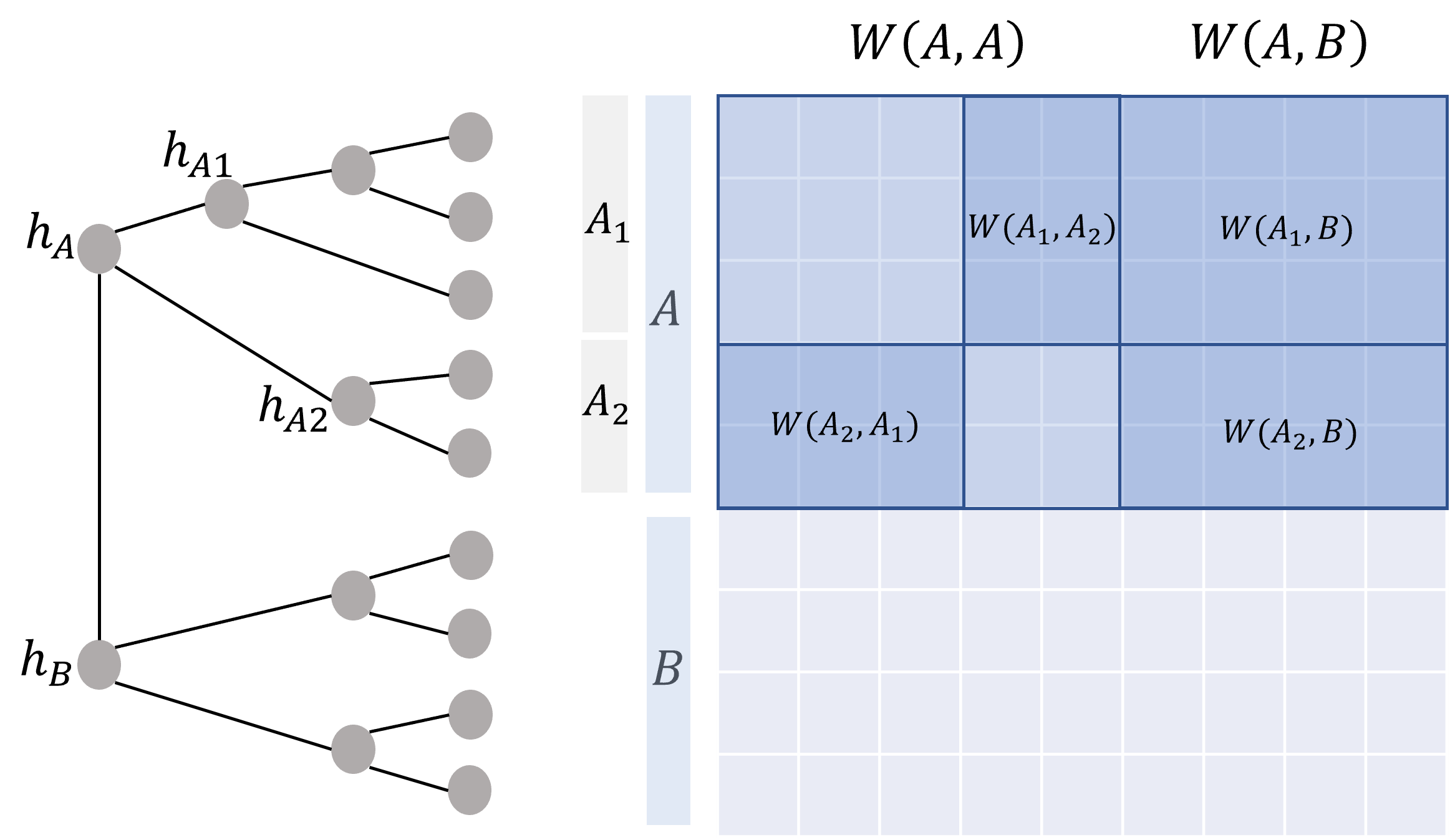}
 \caption{A tree $\T$ and its corresponding similarity matrix $S$. On the right - the matrix block partitions used in the proof of lemma \ref{lem:weight_assignment}}. 
 \label{fig:similarity_structure}
\end{figure}

The four blocks of $S$ satisfy the following two properties with respect to the tree $\T$.
\begin{itemize}[leftmargin=*]
 \item[\RNum{1}] 
 The submatrix $W(A,B)$ is rank-1. 
 
 \item[\RNum{2}] The submatrices $W(A,A),W(B,B)$ satisfy the rank-1 condition with respect to the trees $\T_A$ and $\T_B$ respectively. 
\end{itemize}
Property I follows from the
 fact that $A$ and $B$ are clans of $\mathcal T$, and that by the assumption of the lemma
 $W$ satisfies the rank-1 condition with respect to $\T$. 
  Hence there are vectors $v_A,v_B$ such that 
  \begin{equation}
  W(A,B)= v_A v_B^T
    \label{eq:v_AB}
  \end{equation}
  Property II follows from the above auxiliary
  Lemma \ref{lem:rank_1_submatrices}.

By Property II and the induction hypothesis, we can set weights to the edges of $\T_A$ such that each element $W(i,j)$ where $(i,j) \in A$ equals the product of weights along the path between nodes $x_i$ and $x_j$. Similarly, we can set such weights to $T_B$.
 
We denote by $w(A,h_A) \in \R^{|A|}$ a vector whose entries are the product of weights along the paths between the subset of nodes $A$ and $h_A$. 
Similarly, let  $w(B,h_B) \in \R^{|B|}$ be denote the product of weights between the subset $B$ and $h_B$. 
Our proof is based on the following auxiliary lemma. 
\begin{lemma}\label{lem:rank_1_aux_lemma}
The vectors $w(A,h_A)$ and $w(B,h_B)$ are proportional to $v_A$ and $v_B$ 
defined in Eq. \eqref{eq:v_AB}. Namely,  
\[
w(A,h_A) = c_A v_A, \qquad  w(B,h_B) = c_B v_B ,
\]
for some constants $c_A$ and $c_B$.
\end{lemma}
Lemma \ref{lem:rank_1_aux_lemma} relates properties \RNum{1} and \RNum{2}. For an arbitrary partition $A,B$, the singular vectors of the rank-1 submatrix $W(A,B)$ are proportional to the weights $w(A,h_A)$ and $w(A,h_B)$. 
We first complete the proof of Lemma \ref{lem:weight_assignment} assuming Lemma \ref{lem:rank_1_aux_lemma} is correct. We set the weight  $w(h_A,h_B)=c_Ac_B$. Thus,
\[
w(A,h_A)w(h_A,h_B)w(h_B,B) = w(A,h_A)c_Ac_Bw(h_B,B) = v_A v_B^T = W(A,B).
\]
Finally, we prove that the product $c_Ac_B$ is smaller than 1. Consider a quartet of nodes $(i,j) \in A$ and $(i',j') \in B$ and consider the determinant of the following submatrix
\[
\bigg| 
\begin{pmatrix}
    W(i,j) & W(i,j') \\
    W(i',j) & W(i',j')
\end{pmatrix}
\bigg| = 
\bigg|
\begin{pmatrix}
w(i,h_A)w(h_A,j) & w(i,h_A)c_Ac_Bw(h_B,j') \\
w(j,h_A)c_Ac_Bw(h_B,i') & w(i',h_B)w(h_B,j')    
\end{pmatrix}
\bigg| = W(i,j)W(i',j')(1-(c_Ac_B)^2).
\]
By the definition of the rank-1 matrix, we have that the determinant is positive, which implies that $c_Ac_B<1$.
Together with property \RNum{2}, we found a set of weights such that $W(i,j) = w(x_i,x_j)$ for all pairs $(i,j)$ which proves Lemma \ref{lem:weight_assignment}.

Let us now prove auxiliary Lemma \ref{lem:rank_1_aux_lemma}. Let $h_{A1},h_{A2}$ and $h_{B1},h_{B2}$ be two pairs of nodes adjacent to $h_A$ and $h_B$, respectively. We denote by $A_1,A_2$ a partition of $A$ according to the nodes closer to $h_{A1}$ and $h_{A2}$, see illustration in Figure \ref{fig:similarity_structure}.
Consider the submatrix $W(A_1,A_2)$.
By property \RNum{2}, we can set weights to $\T_A$ such that,
\begin{equation}\label{eq:s_a_1_a_2}
W(A_1,A_2) = w(A_1,h_A)w(h_A,A_2) = w(A_1,h_{A1})w(h_{A1},h_A)w(h_{A},h_{A2})w(h_{A2},A_2).
\end{equation}

Note that we have a degree of freedom in the choice of $w(h_{A1},h_A)$ and $w(h_{A},h_{A2})$ since multiplying one and dividing the other by a constant does not change the product of weights between any pair of nodes in $A$. We will use this degree of freedom shortly.

Next, consider three submatrices: $W(A_1,A_2),W(A_1,B)$ and their concatenation $W(A_1, A_2 \cup B)$. All three submatrices are rank-1, 
$W(A_1,A_2)$ by Eq. \eqref{eq:s_a_1_a_2}, and $W(A_1,B)$ and $W(A_1, A_2 \cup B)$ since the matrices $W(A,A)$ and $W$ satisfy the rank-1 condition with respect to $\T_A$ and $\T$, respectively. Since $W(A_1, A_2 \cup B)$ is a concatenation of $W(A_1,A_2)$ and $W(A_1,B)$, it implies that all three submatrices have the same left singular vector. By Eq. \eqref{eq:s_a_1_a_2} this singular vector is proportional to $w(A_1,h_{A1})$.
We can apply a similar reasoning with the three matrices $W(A_2,A_1),W(A_2,B)$ and $W(A_2,A_1 \cup B)$ to conclude that the left singular vector of all three matrices is proportional to $w(A_2,h_{A2})$. We thus conclude that there are two constants $c_{A1}$ and $c_{A2}$ such that the left singular vector of $W(A_1,B)$ is equal to $c_{A1} w(A_1,h_{A1})$ and the left singular vector of $W(A_2,B)$ is equal to $c_{A2} w(A_2,h_{A2})$. 

The submatrix $W(A,B)$ is a concatenation of $W(A_1,B)$ and $W(A_2,B)$. Its left singular vector $v_A$ is thus a concatenation of $c_{A1} w(A_1,h_{A1})$ and $c_{A2} w(A_2,h_A)$.
Recall that we have a degree of freedom for the choice of $w(h_{A1},h_A)$ and $w(h_{A2},h_A)$. We set the two weights such that 
\[
w(h_{A1},h_A) c_{A1} = w(h_{A2},h_{A2}) c_{A2} \equiv c_A.
\]
With these weights, $v_A$ is a concatenation of $c_A w(A_1,h_A)$ and $c_A w(A_2,h_B)$, which equals $c_A w(A,h_A)$.
A similar reasoning yields that $v_B = c_B w(B,h_B)$ which concludes the proof.  

\end{proof}

\section{Proof of Theorem \ref{thm:partitioning_consistency}}
\label{sec:proof_theorem_partition}

\subsection{Proof of Lemma \ref{lem:population_mean}: $\E[S^{g}]$ satisfies the rank-1 and quartet conditions.}
\begin{proof}[Proof of Lemma \ref{lem:population_mean}]

Let $S(A,B)$ and $S^{g}(A,B)$ denote two similarity matrices between nodes in a subset $A$ and nodes in the complementary subset $B$ in the species tree and gene tree, respectively. 
In the following proof we first assume that $A$ and $B$ are clans, 
and hence by Lemma 3.1 in \cite{snj}  $S(A,B)$ is rank-1. We prove that in this case, $\E[S^{g}(A,B)]$ is also rank-1. Then, we show that if $A$ is not a clan, then $\E[S^{g}(A,B)]$ is not rank-1.

Assume $A$ and $B$ are clans obtained by disconnecting the  edge between $h_A$ and its parent node, and let $i \in A, k \in B$ be a pair of nodes with similarity  $S^g(i,k)$ in the gene tree. By Eq. multiplicity of the similarity

\begin{equation}\label{eq:gene_sim}
S^g(i,k) = \exp\big(-D^g(i,k)\big) = \exp\big(-D(i,k) - 2\Delta_{ik}\big) = S(i,k)\exp(-2\Delta_{ik}).
\end{equation}

Since $k$ is external to the clan $A$, all pairs that include a node $i \in A$ and node $k$ share the same common ancestor, i.e. $h_{i,k} =  h_{A \cup \{k\}}$ for any $i \in A$. Thus, by property \ref{prop_1}, the distributions of $\Delta_{i,k}$, and hence the expectation $\E[\exp(-2\Delta_{i,k})]$ are the same for any node $i\in A$. 

Let $S(A,k)$ be a single column in $S(A,B)$ that contains the similarity values between $k$ and the subset $A$ in the species tree, and let $\E[S^g(A,k)]$ denote the expectation of the corresponding column in the gene tree.
Eq. \eqref{eq:gene_sim} implies that $\E[S^g(A,k)]$ is equal, up to a multiplicative factor, to $S(A,k)$.
The same holds for any node $k' \in B$.
Since $S(A,B)$ is rank 1, then $\E[S^g(A,B)]$ is also rank-1. 

\begin{figure}[t]
 \centering
 \includegraphics[width = 0.4\linewidth]{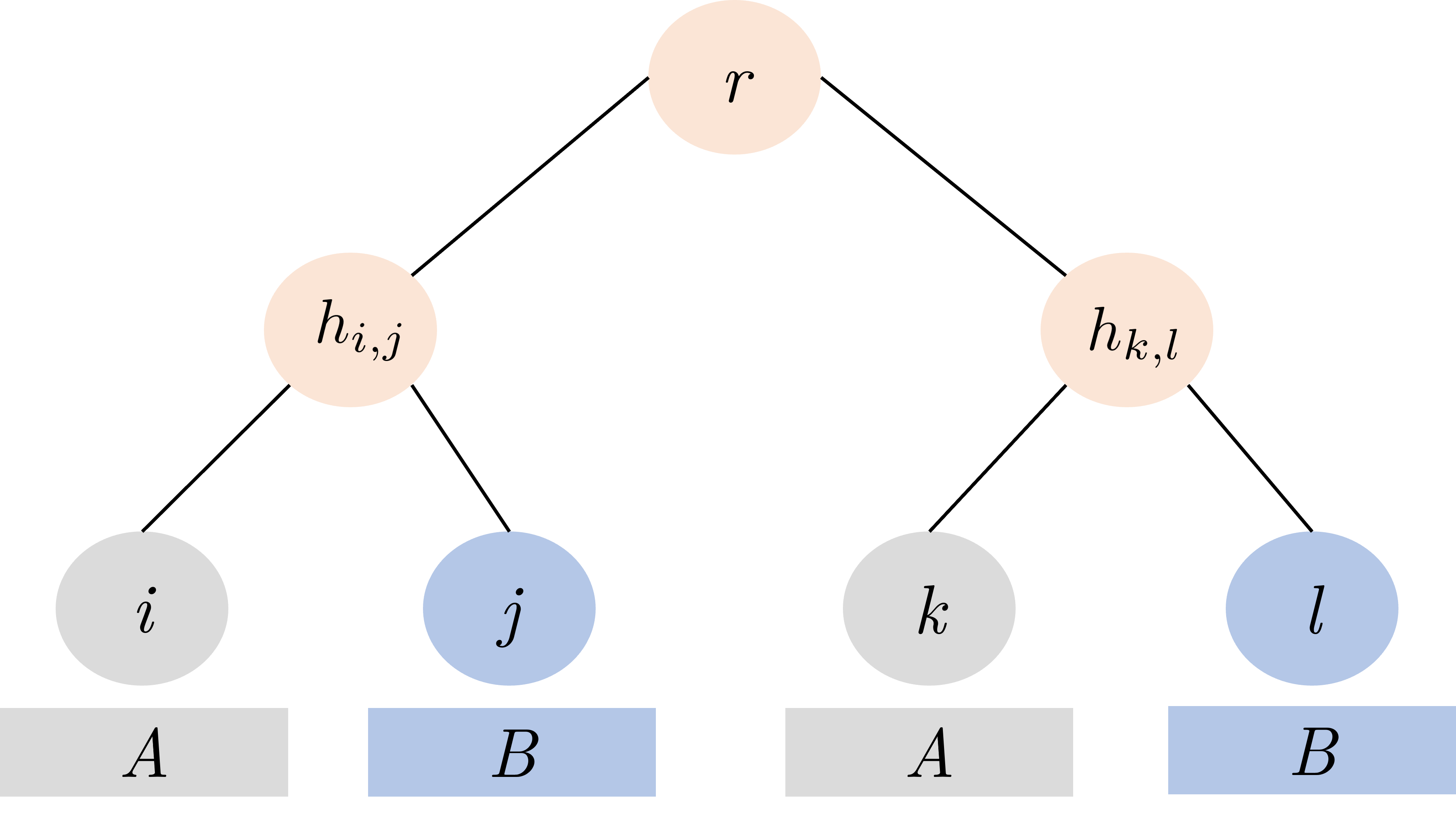}
 \hspace{2em}
 \includegraphics[width = 0.4\linewidth]{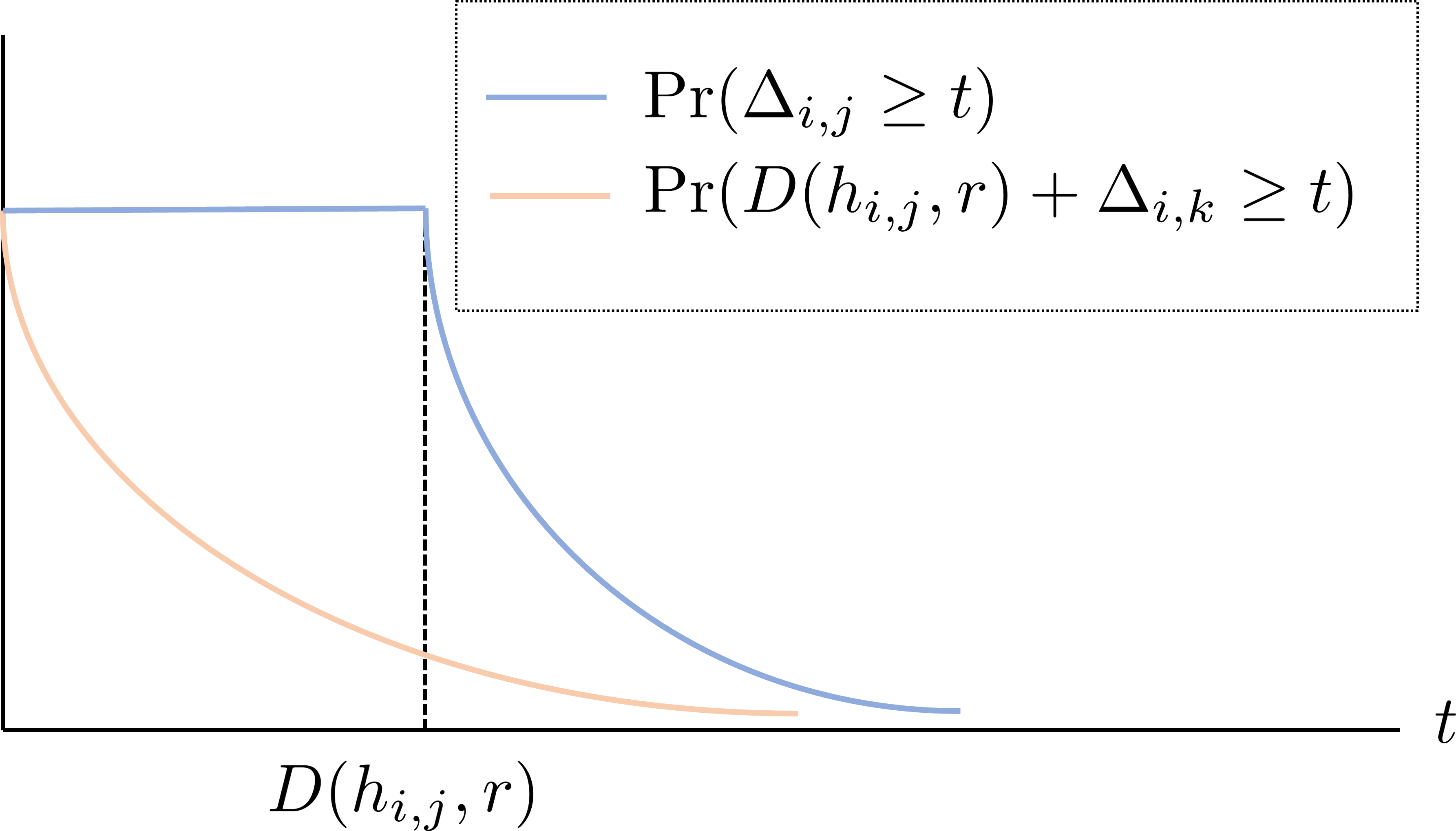}
 \caption{
  Left: Illustration of a quartet tree. For part (b) of our proof, we show that if $A$ and $B$ are not clans, $S(A,B)$ is not rank 1. Right: key property \RNum{2}: The random variable $\Delta_{i,k} + D(h_{i,j},r)$ stochastically dominates $\Delta_{i,j}$.} 
 \label{fig:quartet}
\end{figure}

Next, consider the case where the subset $A$ is not a clan in the species tree.
There are at least two pairs $(i,k) \in A$ and $(j,l) \in B$ where $(i,j)$ and $(k,l)$ are siblings in the inner structure of the quartet in the species tree, see illustration in Fig. \ref{fig:quartet}.
We denote by $r$ the MRCA of the quartet. Consider the $2 \times 2$ submatrix $\E[S^g(A,B)]$,

\begin{equation}\label{eq:gene_submatrix}
\E\Big[S^g\big((i,k),(j,l)\big)\Big] = 
\begin{pmatrix}
 \E[S^g(i,j)] & \E[S^g(i,l)] \\
 \E[S^g(k,j)] & \E[S^g(k,l)]
\end{pmatrix}
\end{equation}
Similar to Eq. \eqref{eq:gene_sim}, each element of this matrix can be expressed in the following form:
\[
\E[S^g(i,j)] = S(i,j)\E[\exp(-2\Delta_{ij})].
\]
In addition, by the multiplicativity property of the similarity measure we have
\[
S(i,l) = S(i,h_{i,j})\exp\big(-D(h_{i,j},r)\big)\exp\big(-D(r,h_{k,l})\big)S(h_{k,l},l), 
\]
\[
S(k,j) = S(k,h_{k,l})\exp\big(-D(h_{k,l},r)\big)\exp\big(-D(r,h_{i,j})\big)S(h_{i,j},j).
\]
Thus,
\[
S(i,l)S(k,j) = S(i,j)S(k,l)\exp(-2D(h_{i,j},r))\exp(-2D(h_{k,l}, r))
\]
The determinant of the submatrix in Eq. \eqref{eq:gene_submatrix} equals,
\begin{align}
&\E[S^g(i,j)]\E[S^g(k,l)]-\E[S^g(i,l)]\E[S^g(k,j)] \notag \\
&= S(i,j)S(k,l) \Big( \E[\exp\big(-2\Delta_{i,j}\big)]\E[\exp\big(-2\Delta_{k,l}\big)] \notag \\ 
&- \E[\exp\big(-2(D(h_{i,j},r)+\Delta_{i,l})\big)]\E[\exp\big(-2(D(h_{k,l},r)+\Delta_{k,j})\big)] \Big)
\end{align}
Consider the triplet $(i,j,l)$. The root node $r$ is an ancestor of $h_{i,j}$ and hence by property \ref{prop_2} it holds that $D(h_{i,j},r)+\Delta_{i,l}$ stochastically dominates $\Delta_{i,j}$. Since $\exp(-2 d)$ is a monotonically decreasing function in $d$, then
\[
\E[\exp(-2(D(h_{i,j},r)+\Delta_{i,l}))] < \E[\exp(-2\Delta_{ij})].
\]
Similarly,  
\[
\E[\exp(-2(D(h_{k,l},r)+\Delta_{k,j}))] < \E[\exp(-2\Delta_{k,l})].
\]
Hence, the determinant of the submatrix is positive, which implies it is rank-2. Thus, $\E[S^g(A,B)]$ cannot be rank-1. 

To prove that the quartet condition holds, it suffices to show that the this condition is satisfied in any similarity matrix of a tree, whose weights are between $0$ and $1$. 
Consider a quartet $Q = ((i,j),(k,l))$ such that $(i,k)$ and $(j,l)$ are siblings in $\T_Q$, see Fig. \ref{fig:quartet}. 
We denote by $S(h_{i,j},h_{k,l})$ the similarity between the ancestor $h_{i,j}$ of $(i,j)$ and the ancestor $h_{k,l}$ of $(k,l)$.
The determinant $|S((i,j),(k,l)|$ is equal to,
\[
|S((i,j),(k,l)| = S(i,k)S(j,l) - S(i,l)S(j,k) = S(i,k)S(j,l)(1-S(h_{i,j},h_{k,l})^2).
\]
Since $0<S(h_{i,j},h_{k,l})<1$, the quartet is positive. 
\end{proof}

%Let $\bar S$ denote the average of the gene similarity matrices and $\E[\bar S]$ denote its population mean, such that
%\[
%\E[\bar S] = \frac1K \sum_{k=1}^K \E[S^{g}] = \E[S^{g}]
%\]
\subsection{Proof of Theorem \ref{thm:partitioning_consistency}}

%Let $\E[S^{g}]$ denote the population mean of the gene similarity matrix. 
%The theorem is proved in two steps. First, in Lemma \ref{lem:population_mean} we prove that $\E[S^{g}]$ satisfies the rank-1 condition with respect to $\T$. This immediately implies via Theorem \ref{thm:snj_stdr} that partitioning the species according to the Fiedler vector of a graph whose weight matrix is equal to $\E[S^{g}]$ yields two clans. Then, we show that the same holds with a sufficiently small perturbation to the graph weights. 

%\textcolor{red}{Boaz: Cannot follow the proof below. First of all what is $L^{g}$? If the idea is the graph Laplacian corresponding to $\bar S$, then it should be denoted $\bar L$.
%In addition, expected graph Laplacian 
%$\E[S^{g}]$ does not make sense. }

\begin{proof}[Proof of Theorem \ref{thm:partitioning_consistency}]
By 
Lemma \ref{lem:population_mean},  $\E[S^{g}]$ satisfies the rank-1 condition with respect to $\T$, namely $\E[S^{g}] \in \WT$.
Let $v$ denote the Fiedler vector of a graph whose weight matrix is equal to $\E[S^{g}]$. %yields two clans.
%This immediately implies via 
Since $\E[S^{g}] \in \WT$, Theorem \ref{thm:snj_stdr} implies that partitioning the species according to $v$ yields two clans.

Let $\bar L$ be the Laplacian of a graph whose weight matrix equals the sample mean $\bar S = \frac1K \sum_{g=1}^K S^{g}$. We denote by $\hat v$ the corresponding Fiedler vector, and by $0 = \gamma_1<\gamma_2<\gamma_3$ the smallest eigenvalues of $\bar L$.  %Recall that $\gamma_2$ 
%Let $v$ and $\hat v$ denote the Fiedler vectors whose weights matrix is equal to $\E[S^{g}]$ and $\bar S$, respectively. 
By the Davis-Kahan theorem \cite{yu2015useful}, %we have 
\begin{equation}\label{eq:davis_kahan}
\|\hat v-v\| \leq 2^{3/2}\frac{\|\bar L- \E[\bar L]\|}{\min(\gamma_2,\gamma_3-\gamma_2)}. 
\end{equation}
By the law of large numbers $\bar S$ converges to $\E[S^{g}]$, and hence $\bar L$ to $\E[\bar L]$. %We also assume that $\min_i |(v_g)_i|$ is bounded away from zero. 
Thus, the Davis-Kahan bound (r.h.s of Eq. \eqref{eq:davis_kahan}) converges to zero for $K \to \infty$. Since by assumption $\min_j |v(j)| > 0$, the sign pattern of $\hat v$ is the same as $v$ once the Davis-Kahan bound falls below $\min_j |v(j)|$. Thus, partitioning the terminal nodes according to the sign pattern of $\hat v$ yields the same subsets as that of $v$, and is guaranteed to yield two clans in the species tree.
\end{proof}

\section{Evaluating the partitioning step on the $50$-species and $200$-species datasets}
\label{sec:200_50_partitions}

Here, we evaluate the accuracy of the partitioning step of \texttt{SDSR} as a function of the number of gene alignments $K$ for the $50$-species and $200$-species datasets. 
Recall that the parameter $\beta$ sets the minimum partition in the constrained K-means.
For both the $50$-species and the $200$-species datasets, $\beta$ was set to $0.05$. The threshold parameter $\tau$ was set to $50$ for the $200$-species datasets and to $25$ for the $50$-species datasets. 
We compare the outcome of the first partition to the most similar partition in the species tree. Figure \ref{fig:200_Species_Datasets_Partition_based_similarity} shows the Rand Index score as a function of the number of genes.
The points mark the median of the 50 datasets, and the line is the range between the $25\%$ and $75\%$ quantiles.
The results show that a small number of genes is sufficient to obtain an accurate partition for these datasets.

\begin{figure}[t]
\centering
\makebox[\textwidth][c]{%
 \includegraphics[width=0.9\textwidth]{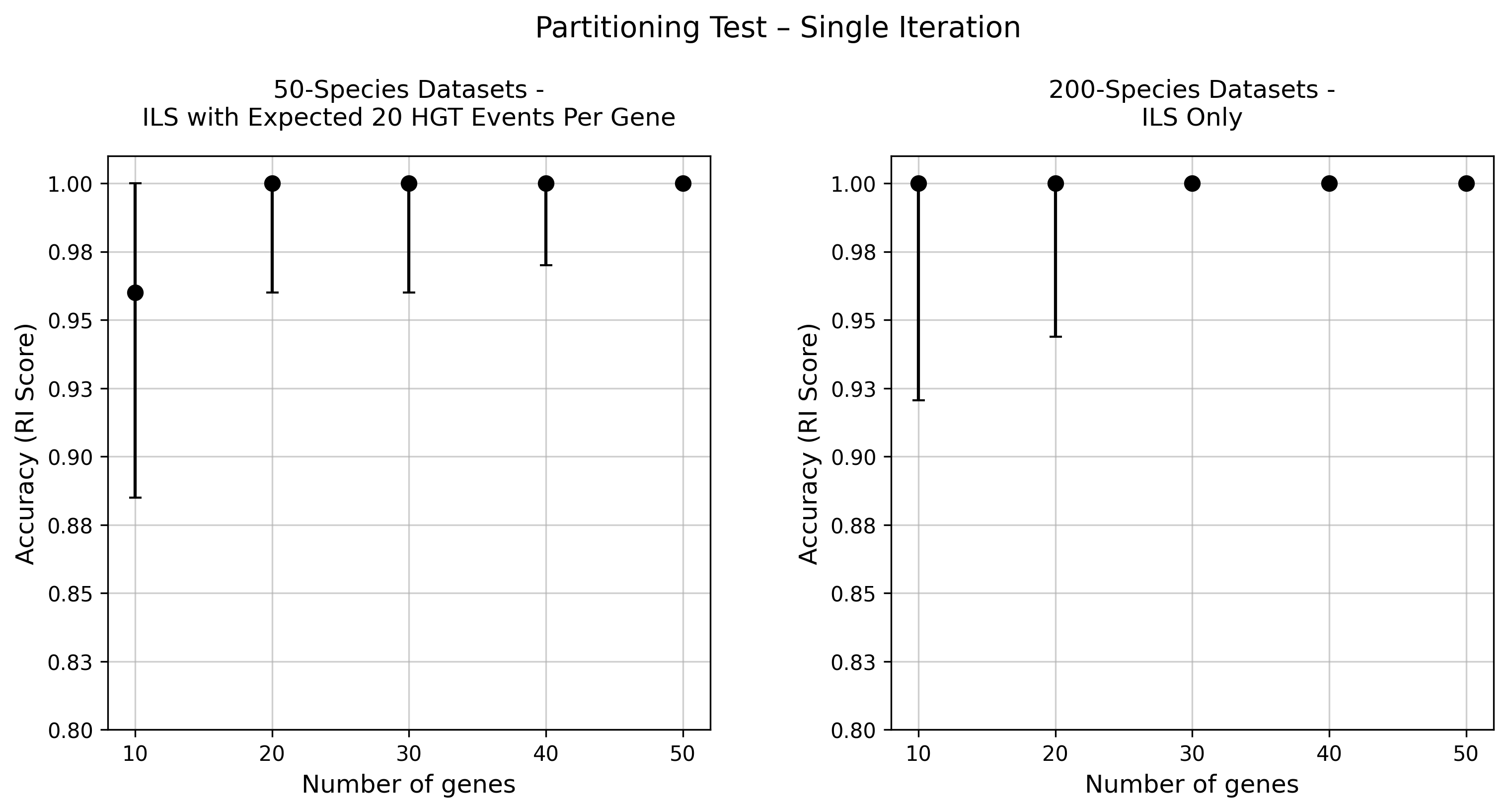}%
}
\caption{
Accuracy of the first partition 
as a function of the number of genes, 
on the $50$-species datasets (left) and the $200$\bl{-}species datasets (right). The accuracy is measured by the Rand Index similarity score between the partition obtained by the Fiedler vector and the most similar partition in the ground-truth species tree.}
\label{fig:200_Species_Datasets_Partition_based_similarity}
\end{figure}

\section{Test partitioning via  average distance matrix}
\label{sec:distance_partition}

As described in Section \ref{sec:method}, the \texttt{SDSR} partitioning approach computes a similarity matrix from each gene alignment. Here, we explore an alternative approach that is based on an average distance matrix.
Let \( D^g \) denote a pairwise distance matrix estimated from the $g$-th gene alignment, and let $\bar D$ denote the average distance matrix:
\[
\bar{D} = \frac{1}{K} \sum_{g=1}^{K} D^g.
\]
We compute a single similarity matrix via,
\[
S(i,j) = \exp(-\bar D(i,j)), 
\] 
and compute its corresponding Laplacian. 
Partitioning is then performed based on the Fiedler vector, as described in Step 1a of Section \ref{sec:method}.
To evaluate the performance of this approach, we conduct experiments on the $50$ and $200$ species datasets. Gene lengths and parameter values for $\beta$ and $\tau$ are as described in Section \ref{sec:200_50_partitions}.
Figure \ref{fig:200_Species_Datasets_Partition_based_distance} presents the Rand Index score as a function of the number of genes.
The points mark the median of the 50 datasets, and the line represents the range between the $25\%$ and $75\%$ quantiles. The results show that partitions obtained via the distance-based approach are less accurate than those obtained using the averaged Laplacian approach. 

\begin{figure}[t]
\centering
\makebox[\textwidth][c]{%
 \includegraphics[width=0.95\textwidth]{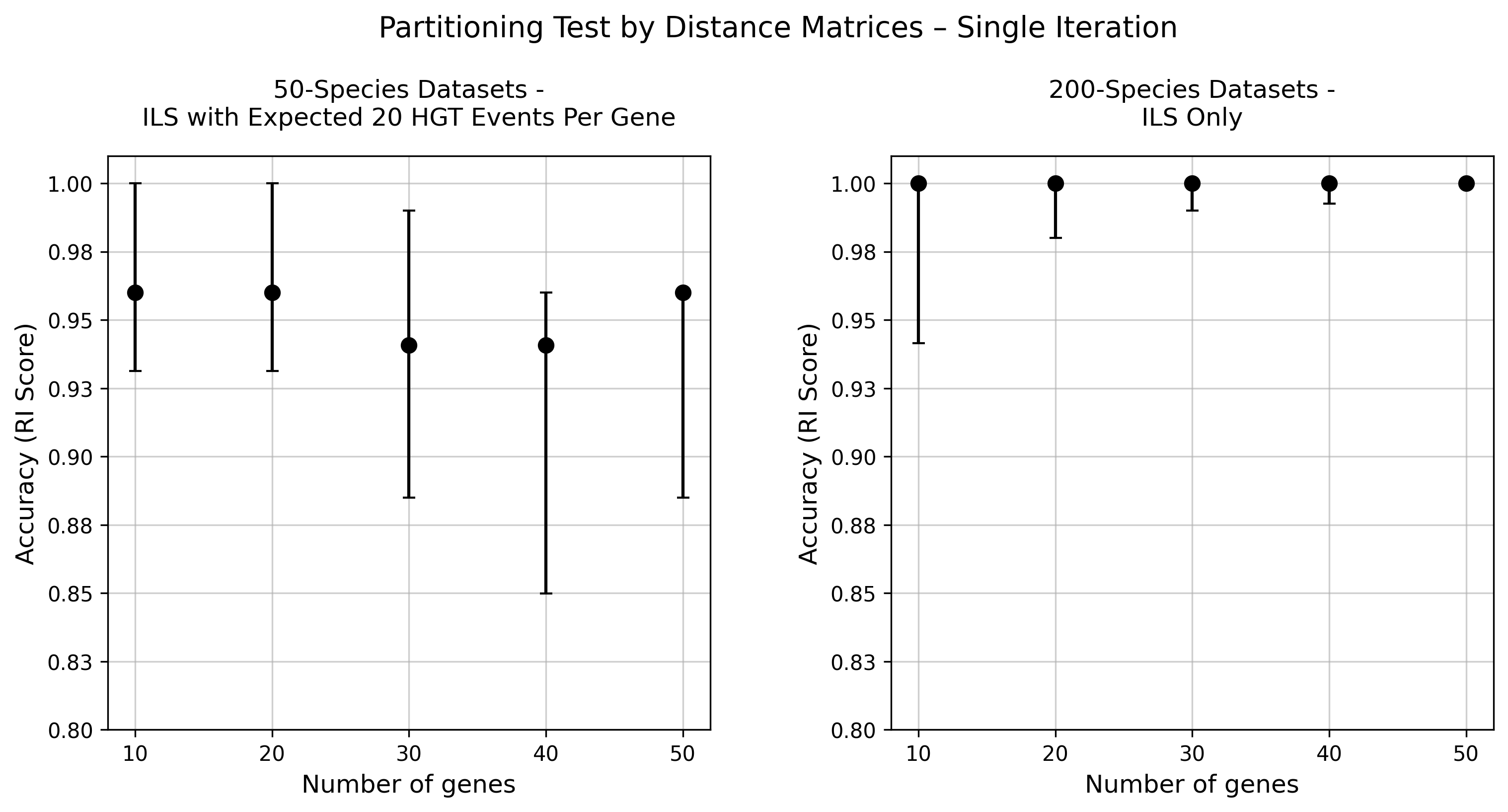}%
}
\caption{Accuracy of the first partition, based on distance matrices,  
as a function of the number of genes. 
(left) $50$\bl{-}species dataset; (right) $200$\bl{-}species dataset. The accuracy is measured by the Rand Index similarity score between the partition obtained by the Fiedler vector and the most similar partition in the ground-truth species tree.}
\label{fig:200_Species_Datasets_Partition_based_distance}
\end{figure}

\section{Reconstruction accuracy with a single \texttt{SDSR} iteration}\label{sec:partition_experiment_single_iteretion}

We compared the accuracy and runtime of ASTRAL and CA-ML when used as stand-alone methods and when used as subroutines of \texttt{SDSR} with a single iteration , single CPU, and $\beta = 0.05$.
The experiment was conducted on the $50$ and $200$ species datasets. 
Figure \ref{fig:50_species} shows the accuracy and runtime for all methods as a function of the number of genes for the $50$ species datasets. The markers represent the median, and the lines show the range between the $25\%$ and $75\%$ quantiles. 
For a large number of genes, there is almost no difference between the accuracy of ASTRAL and \texttt{SDSR} + ASTRAL, and between CA-ML and \texttt{SDSR} + CA-ML. 
However, the spectral approach demonstrates a significant reduction in runtime. For instance, for the ILS-only datasets with $250$ genes, the processing time of \texttt{SDSR} + CA-ML was around $40\%$ of that of CA-ML.

\begin{figure}[t]
 \centering
 \includegraphics[width=0.9\textwidth]{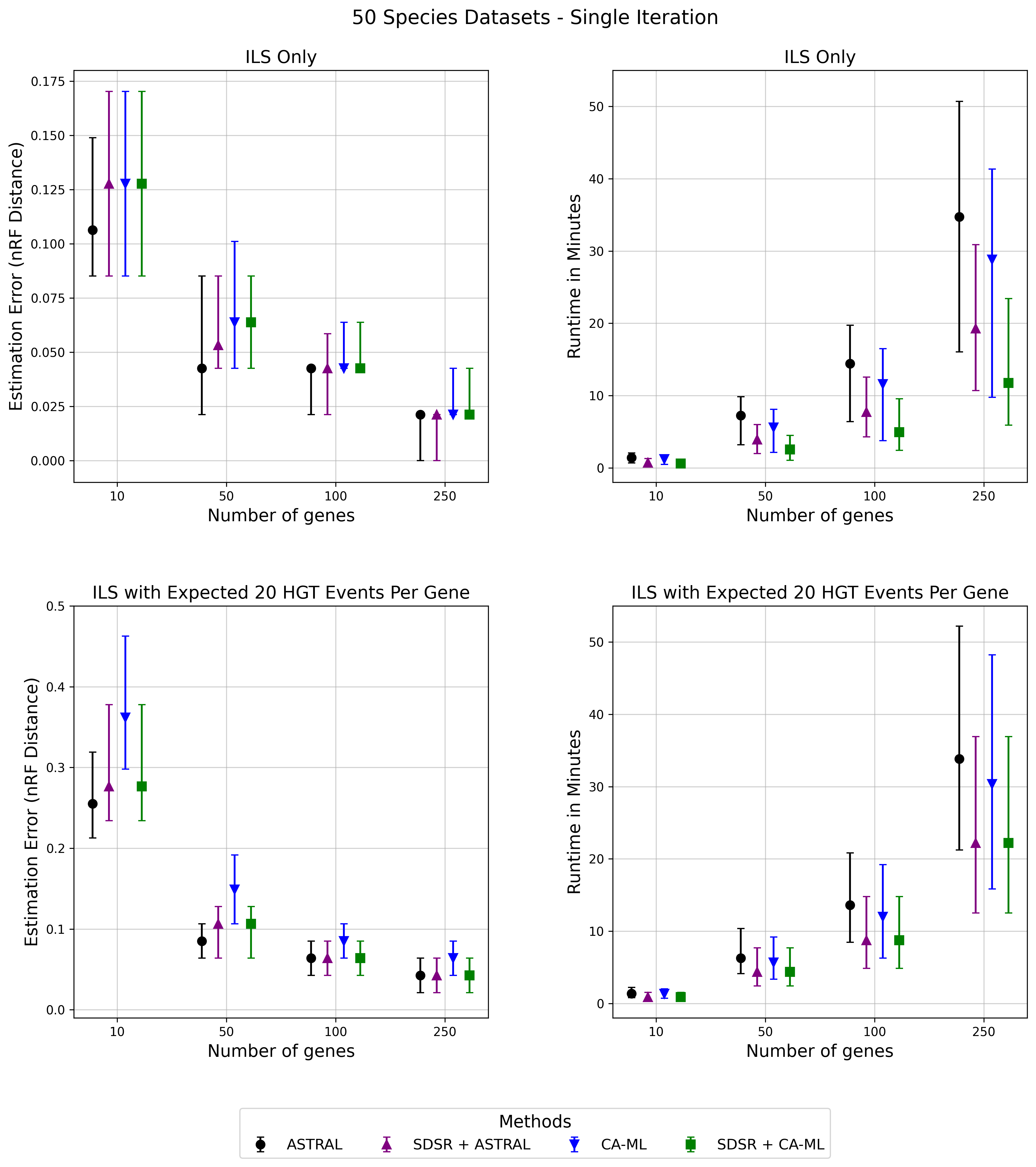}

 \caption{Results for the $50$ species datasets with a single \texttt{SDSR} iteration: The left panels show the nRF distance between the estimated and ground-truth species trees as a function of the number of genes, for the dataset with only ILS (upper panel) and ILS+HGT (bottom panel). The right panels show the runtime in minutes.}
 \label{fig:50_species}
\end{figure}

\begin{figure}[t]
 \centering
 \includegraphics[width=0.9\textwidth]{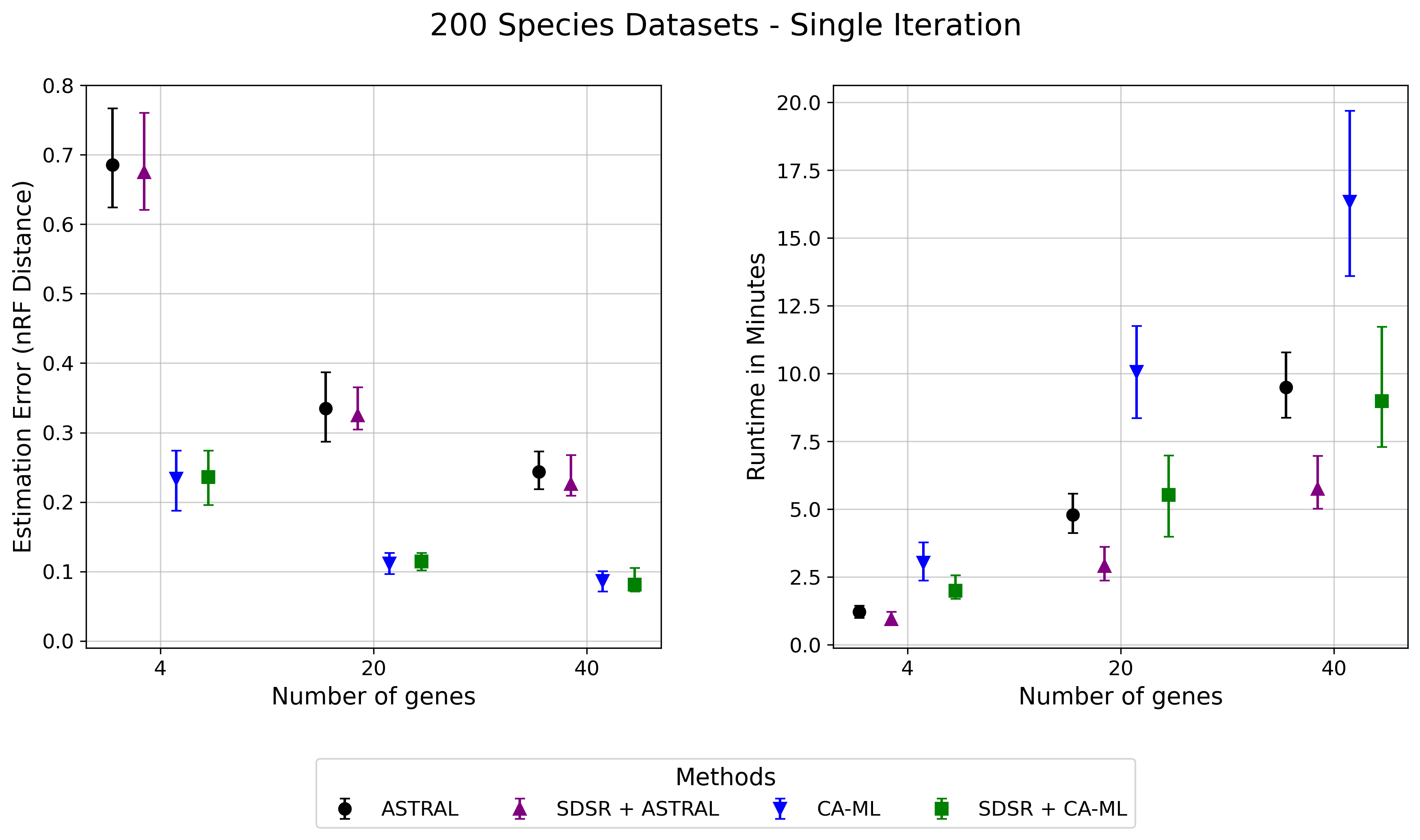}
 \caption{Results for the $200$ species datasets with a single \texttt{SDSR} iteration: (left panel) nRF distance between the estimated and ground-truth species trees as a function of the number of genes; (right panel) runtime in minutes. 
 }
 \label{fig:200_species}
\end{figure}

\section{Reconstruction accuracy with different threshold values}\label{sec:multiple_partitions}

Here, we measure the effect of $\tau$, the threshold parameter that upper bounds the size of the tree partitions, on the accuracy and runtime. The experiment was conducted using the $200$ species datasets.
We applied CA-ML as a stand-alone method and as a subroutine within \texttt{SDSR} on a single CPU, using $\beta = 0.05$, and various values of $\tau$. 
Figure \ref{fig:recursion level} shows the runtime and accuracy vs. the number of genes.
The lines mark the range between the $25\%$ and $75\%$ quantiles, and the markers represent the median.
For lower values of $\tau$, the spectral approach has a significantly lower runtime, with a similar accuracy.

\begin{figure}[t]
 \centering
 \includegraphics[width=0.9\textwidth]{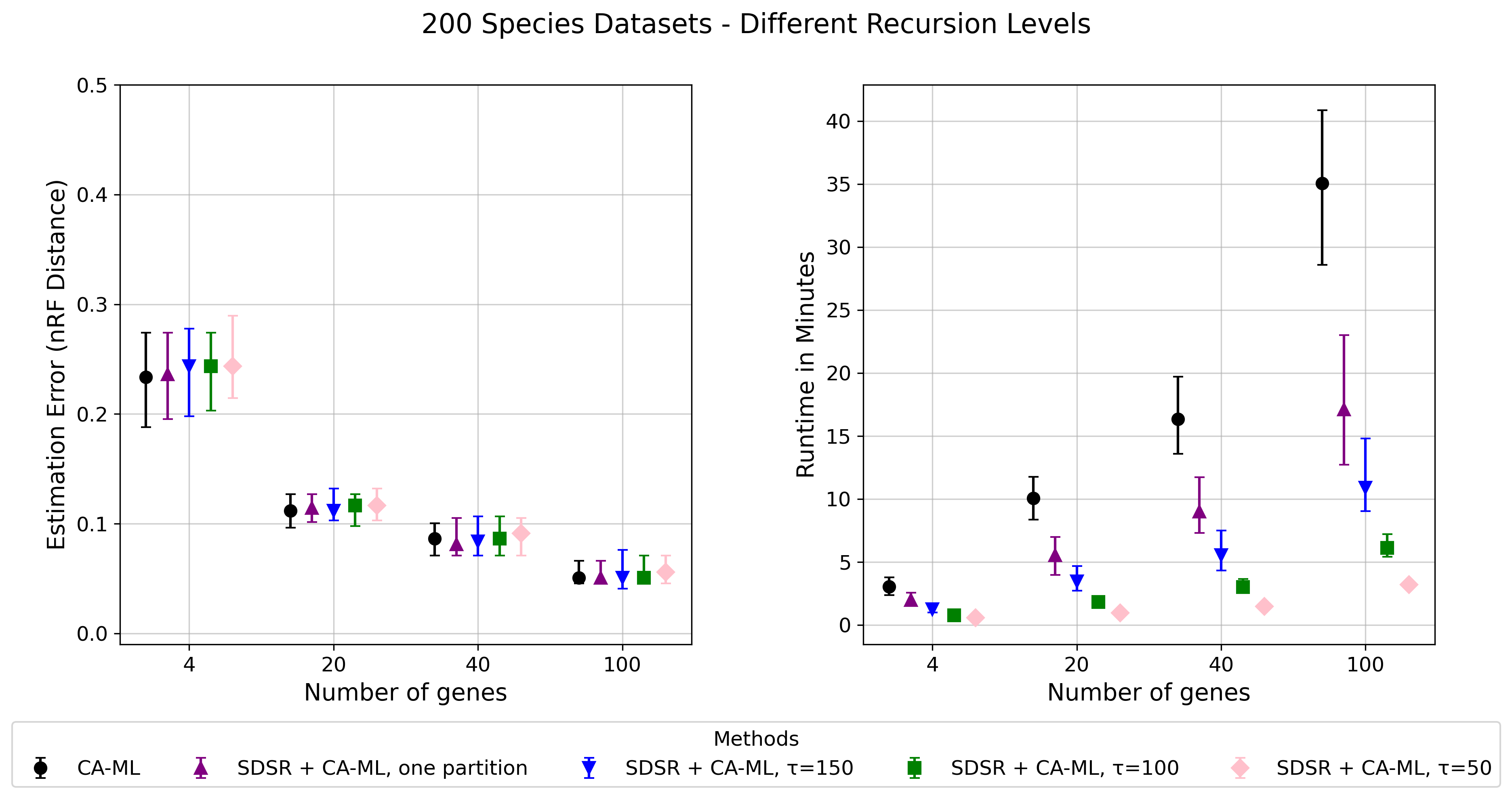}
 \caption{Results for the $200$ species datasets - multiple \texttt{SDSR} iterations on a single CPU: (left panel) nRF distance between the estimated and ground-truth species trees as a function of the number of genes; (right panel) runtime in minutes. 
 }
 \label{fig:recursion level}
\end{figure}

\section{Proofs for the finite genes guarantee}\label{sec:app_finite_gene}

This section is dedicated to the proof of Theorem~\ref{thm:finite_gene_guarantee_v2}, and it is organized as follows. In \ref{subsec:v2_proof} we prove Theorem \ref{thm:finite_gene_guarantee_v2} based on Lemma \ref{lem:finite_gene_guarantee}, which proved based on lemmas \ref{lem:the_allowed_error} and \ref{lem:concent_bound} in Subsection \ref{subsec:finite_gene_proof}. Then Lemma \ref{lem:the_allowed_error} and Lemma \ref{lem:concent_bound} proved in subsections \ref{subsec:allowed_error_justification} and \ref{subsec:conct_ineq_proof}, respectively. Finally, several auxiliary lemmas used in \ref{subsec:conct_ineq_proof} were proved in \ref{subsec:proof_auxiliary}.

\subsection{Proof of Theorem \ref{thm:finite_gene_guarantee_v2}}\label{subsec:v2_proof}

We first present two auxiliary results, Lemma \ref{lem:finite_gene_guarantee} and Lemma \ref{lem:spec_norm_similarity_bnd}. To state Lemma \ref{lem:finite_gene_guarantee}, we first
introduce \textit{the allowed error} $\Delta(\mathcal T_E)$, where $\mathcal T_E$ is a tree with the same topology as $\mathcal T$ and with similarity $\mathbb{E}[S^g]$ as discussed in Subsection \ref{subsec:finite_gene_garantee}.  
The quantity $\Delta(\mathcal T_E)$ is 
given by
\begin{equation}\label{eq:allowed_error_quantity}
        \Delta(\mathcal{T}_E)=\frac{\sqrt{m} \cdot S_{\min}}{\sqrt{\eta}2^{3/2}\left(\sqrt{m}+1\right)}\min\left\{ 1,\frac{1}{1+\eta}\left(\frac{1}{s_r^2}-1\right)\right\},
\end{equation}
where the quantities $\eta,S_{\min}$ and $s_r$
were all defined in Subsection \ref{subsec:finite_gene_garantee}. 

The next lemma, proved in Subsection \ref{subsec:finite_gene_proof}, states that if the number of genes is larger than a quantity that depends on $\Delta(\mathcal T_E)$, then spectral partition yields two clans of the species tree. 
In what follows, $\|S\|$ is the spectral norm of the matrix $S$.
\begin{lemma}\label{lem:finite_gene_guarantee}
    Let $\mathcal T$ be an ultrametric species tree with $m$ terminal nodes and similarity matrix $S$. 
    Let $h_L,h_R$ be the two children of the root $r$ of $\mathcal T$ and let $C_L,C_R$ be their observed descendant species, respectively. 
    Let $\{\mathcal{T}^g\}_{g=1}^K$ be $K$ i.i.d. gene trees generated according to the MSC+GTR model,  
    and let $\{S^g\}_{g=1}^K$ be their exact similarity matrices. Assume that for some $\epsilon>0$
    \begin{equation}\label{eq:finite_gene}
        K \geq \left(\frac{\|S\|^2}{2\Delta^2(\mathcal{T}_E)}+\frac{2\|S\|}{3\Delta(\mathcal{T}_E)} \right)\cdot \log\left(\frac{2m}{\epsilon}\right),
    \end{equation}    
    then, with probability larger than $1-\epsilon$, partitioning based on the Fiedler vector corresponding to $\bar S=\frac1K \sum_g S^{g}$ yields the two clans $C_L$ and $C_R$ of the species tree $\mathcal T$.
\end{lemma}

The lower bound \eqref{eq:finite_gene} depends on $\|S\|$. 
Next, we upper bound this quantity. 
Since all entries $S(i,j)\in[0,1]$, a trivial bound is $\|S\|\leq m$. 
In the following lemma, proved in Appendix~\ref{subsec:proof_auxiliary}, we present a sharper bound.
To describe it, recall the definition of $\xi$ as the upper bound on the similarity between adjacent nodes.

\begin{lemma}\label{lem:spec_norm_similarity_bnd}
    Let $\mathcal{T}$ be a (rooted or unrooted) binary tree with $m$ terminal nodes, and let $S$ denote its associated similarity matrix.
    Assume that for some constant $\xi<1$, 
    the similarity values
    $\mathcal S(e)$ at all edges $e$ of $\mathcal T$ satisfy $\mathcal S(e) \leq \xi$. 
    Then, 
    \begin{equation}\label{eq:S_norm_bound}
        \|S\|\leq
        \begin{cases}
        3/2 & \xi \leq 0.5 \\
        1 + 2(m-1)^{1 + \log_{2} \xi} &  \xi > 0.5
        \end{cases}    
    \end{equation}
\end{lemma}

With these two results, we now prove Theorem~\ref{thm:finite_gene_guarantee_v2}.

\begin{proof}[Proof of Theorem \ref{thm:finite_gene_guarantee_v2}]
    Assume that the conditions of Theorem \ref{thm:finite_gene_guarantee_v2} hold, and in particular that the number of genes $K$ satisfies Eq. \eqref{eq:RHS_bound}. 
    To prove the corollary, is suffices to show that the right-hand side of \eqref{eq:finite_gene} is smaller than that of \eqref{eq:RHS_bound}. 
    Since the RHS of Eq. \eqref{eq:finite_gene} depends in $\|S\|/\Delta(\mathcal T_E)$, we separately upper bound $\|S\|$ and $\frac{1}{\Delta(T_E)}$.
    
    To bound $\frac{1}{\Delta(\mathcal{T}_E)}$, we take the reciprocal of Eq. 
    \eqref{eq:allowed_error_quantity},  
    \begin{equation} \label{eq:delta_inv_expanded}
        \frac{1}{\Delta(\mathcal{T}_E)} = \frac{\sqrt{\eta}\cdot 2^{3/2}(\sqrt{m}+1)}{\sqrt{m}} \cdot (S_{\min})^{-1} \cdot \max\left\{1, (1+\eta)\cdot \frac{s_r^2}{1 - s_r^2} \right\}.
    \end{equation}
    Next, we bound the maximum term in \eqref{eq:delta_inv_expanded}. Since $\eta > 1$ and $s_r<1$, 
    \begin{equation} \label{eq:max_term_bound}
        \max\left\{1, (1+\eta)\cdot \frac{s_r^2}{1 - s_r^2} \right\} \leq \frac{1+\eta}{1 - s_r^2}.
    \end{equation}
    Now, notice that $\frac{\sqrt{\eta}\cdot 2^{3/2}(\sqrt{m}+1)}{\sqrt{m}} \leq c \sqrt{\eta}$ for some constant $c$. 
    Inserting this inequality and Eq. \eqref{eq:max_term_bound} into Eq. \eqref{eq:delta_inv_expanded} gives
    \begin{equation} \label{eq:delta_inv_final}
        \frac{1}{\Delta(\mathcal{T}_E)} \leq c_1\frac{\eta^{3/2}}{(1 - s_r^2)S_{\min}},
    \end{equation}
    where $c_1$ is some constant. 
    
    We now proceed to bounding the spectral norm  $\|S\|$. 
    By Equation~\eqref{eq:S_norm_bound}
    of Lemma \ref{lem:spec_norm_similarity_bnd}, 
    \begin{equation} \label{eq:S_norm_bound2}
        \|S\| \leq 
        \begin{cases}
            3/2 & \text{if } \xi \leq \tfrac{1}{2}, \\
            c_2 \cdot m^{1 + \log_2 \xi} & \text{if } \xi > \tfrac{1}{2}.
        \end{cases}
    \end{equation}
    where $c_2$ is a universal constant.

    Next, since both  $\|S\|\geq 1$ and $\frac{1}{\Delta(\mathcal{T}_E)}\geq 1$, it follows that $\frac{\|S\|}{\Delta(\mathcal{T}_E)}\leq\frac{\|S\|^2}{\Delta^2(\mathcal{T}_E)}$. Thus, 
    the term in Eq. \eqref{eq:finite_gene} may be bounded as follows, 
    \begin{equation}\label{eq:dominant_term}
        \left(\frac{\|S\|^2}{2\Delta^2(\mathcal{T}_E)}+\frac{2\|S\|}{3\Delta(\mathcal{T}_E)} \right)\leq \frac{7}{6}\frac{\|S\|^2}{\Delta^2(\mathcal{T}_E)}.
    \end{equation}  

    Finally, combining equations  \eqref{eq:delta_inv_final}, \eqref{eq:S_norm_bound2} and \eqref{eq:dominant_term} yields
    \[
    \left(\frac{\|S\|^2}{2\Delta^2(\mathcal{T}_E)}+\frac{2\|S\|}{3\Delta(\mathcal{T}_E)} \right)\cdot \log\left(\frac{2m}{\epsilon}\right)
    \leq
    \begin{cases}
    C_1 \cdot \frac{\eta^{3}}{(1 - s_r^2)^2S^2_{\min}} \cdot \log(\tfrac{m}{\epsilon}) & \text{if } \xi \leq \tfrac{1}{2}, \\
    C_2 \cdot \frac{\eta^{3}}{(1 - s_r^2)^2S^2_{\min}} \cdot m^{2 + 2\log_2 \xi}\cdot\log(\tfrac{m}{\epsilon}) & \text{if } \xi > \tfrac{1}{2}.
    \end{cases}
    \]
    where $C_1, C_2$ are absolute constants. 
    As discussed above, this inequality yields Eq. \eqref{eq:finite_gene}. 
\end{proof}

\subsection{Proof of Lemma \ref{lem:finite_gene_guarantee}}\label{subsec:finite_gene_proof}

The proof relies on the following two auxiliary lemmas. 
The first is a deterministic result showing that if $\|\bar{S} - \mathbb{E}[S^g]\| \le \Delta(\mathcal{T}_E)$, then the Fiedler vector yields a correct partition. 
The second is a concentration inequality on $\|\bar{S} - \mathbb{E}[S^g]\|$. 
Lemma \ref{lem:finite_gene_guarantee} is proved by bounding the probability that $\|\bar{S}-\mathbb{E}[S^g]\|\geq\Delta(\mathcal{T}_E)$.
We first formally state these two auxiliary results.

\begin{lemma}\label{lem:the_allowed_error}
    Let $\T^g$ be a gene tree generated according to the MSC+GTR model with an underlying 
    ultrametric species tree $\cal T$ with $m$ species. 
    Let $\mathcal T_E$ be the tree that corresponds to the matrix $\mathbb{E}[S^g]$, where $S^g$ is the similarity of $\T^g$.
    Let $\Delta(\mathcal{T}_E)$ be the allowed error defined in Eq. \eqref{eq:allowed_error_quantity}. 
    Then, for any symmetric matrix $A$ 
    \begin{equation}
        \label{eq:allowed_error2}
    \|A-\mathbb{E}[S^g]\|\leq \Delta(\mathcal{T}_E)
    \end{equation}
    partitioning the terminal nodes based on the Fiedler vector of the Laplacian of $A$ yields two clans in the species tree $\mathcal T$.
\end{lemma}

The above lemma is a variant of Lemma 5.1 from \cite{aizenbud2021spectral}. 
To apply the proof of \cite{aizenbud2021spectral} to our setting, 
we have to show that if $\mathcal{T}$ is ultrametric, then $\mathcal{T}_E$ is also ultrametric. 
This is established in Subsection \ref{subsec:allowed_error_justification}.

Next, we present the second result. 
Its proof appears in Appendix \ref{subsec:conct_ineq_proof}.
\begin{lemma}[Concentration Inequality]\label{lem:concent_bound}
    Let $\T$ be a species tree with $m$ terminal nodes, and similarity matrix $S$. 
    Let $\{\T^g\}_{g=1}^K$ be $K$ i.i.d. gene trees generated according to the MSC+GTR model,
    and let $\{S^g\}_{g=1}^K$ be their exact similarity matrices.
    Then, the matrix $\bar{S}=\frac{1}{K}\sum_{g=1}^K S^g$ satisfies 
    \begin{equation}
        \label{eq:concentration_bar_S}
        \Pr(\|{\bar{S}-\mathbb{E}[S^g]}\|>t)\leq 2m\exp{\left(-\frac{Kt^{2}/2}{\| S\|^{2}/4+\| S\|t/3}\right)}, 
        \quad \forall t\geq 0.
    \end{equation}
\end{lemma}

We are now ready to prove  Lemma~\ref{lem:finite_gene_guarantee}.
\begin{proof}[Proof of Lemma~\ref{lem:finite_gene_guarantee}]
    Let $\Delta(\mathcal{T}_E)$ be the allowed error, defined in \eqref{eq:allowed_error_quantity}. 
    By Lemma~\ref{lem:concent_bound}, with $t = \Delta(\mathcal{T}_E)$

    \begin{equation}\label{eq:error_probability}
        \Pr\left( \|\bar{S} - \mathbb{E}[S^g]\| \leq  \Delta(\mathcal{T}_E) \right) 
        \geq 1- 2m \exp\left( -\frac{K \Delta^2(\mathcal{T}_E)/2}{\|S\|^2/4 + \|S\| \Delta(\mathcal{T}_E)/3} \right).
    \end{equation}
    By Lemma~\ref{lem:the_allowed_error}, if $\|\bar S-\mathbb{E}[S^g]\|\leq \Delta(\mathcal T_E)$, the spectral partitioning is guaranteed to output a correct partition. 
    It is easy to verify that if the number of genes $K$ satisfies Eq. \eqref{eq:finite_gene}, then the RHS of \eqref{eq:error_probability} is at least $1-\epsilon$, which concludes the proof. 
\end{proof}

\subsection{On the proof of Lemma~\ref{lem:the_allowed_error}}\label{subsec:allowed_error_justification}

Here we address the difference between Lemma \ref{lem:the_allowed_error} and its reference Lemma 5.1 in \cite{aizenbud2021spectral}. 
Both results state a condition under which the sign pattern of the Fiedler vector yields a correct partition of the leaves into two clans. 
Specifically, the theorems require that the input matrix for the algorithm is not too far in spectral norm from a reference matrix.

The key difference lies in the assumptions made about the trees that correspond to the reference matrices.
In \cite{aizenbud2021spectral}, the underlying tree $\T_1$ is assumed to be ultrametric, and its similarity is the reference matrix.
In contrast, Lemma \ref{lem:the_allowed_error} assumes an ultrametric species tree $\T$, but the reference matrix is $\mathbb E[S^g]$, which corresponds to a different tree $\T_E$.
To apply the original proof from \cite{aizenbud2021spectral}, it suffices to show that $\T_E$ is also ultrametric.
This is formalized in the following lemma.

\begin{lemma}\label{lem:S_and_ESg}
    Let $S$ be the similarity matrix of an ultrametric species tree $\mathcal{T}$. 
    Let $S^g$ be the similarity matrix of a gene tree $\T^g$ generated according to the MSC+GTR model. 
    Then, the tree $\mathcal{T}_E$ with similarity $\mathbb{E}[S^g]$ is also ultrametric.
\end{lemma}

\begin{proof}[Proof of Lemma \ref{lem:S_and_ESg}]
    Recall that an ultrametric tree is a tree where all the leaves are at the same log-det distance from the root. 
    Equivalently, in an ultrametric tree, the similarity between every leaf and the root is equal.
    We start with a property of ultrametric trees.
    
    Let $\T$ be a rooted tree with similarity $S$, and $C_l,C_r$ be the two clans induced by the removal of the root.
    Then, easy to see that 
    \begin{equation}\label{eq:ultrametric}
    \T \text{ is ultrametric} \iff 
    \forall i\in C_l, j\in C_r: \quad S_{ij}=s_0   ,
    \end{equation}
    where $s_0$ is some constant.
    
    Next, let $\T$ be a species tree, $S^g$ a gene similarity matrix generated according to the MSC+GTR model, and $\T_E$ the tree corresponding to the matrix $\mathbb E[S^g]$. 
    Recall $\T_E$ and $\T$ share the same topology, therefore $C_l,C_r$ are clans in both trees.
    By Eq. \eqref{eq:ultrametric}, it suffices to show that $\mathbb E[S^g_{ij}]=s_0$ for all $i\in C_l, j\in C_r$.
    
    By Eq. \eqref{eq:ultrametric}, since $\T$ is ultrametric, $S_{ij}=s_1$ for all $i \in C_l$ and $j \in C_r$, for some constant $s_1$.
    By the gene similarity under the GTR, Eq. \eqref{eq:GTR_similarity}, for any $i\in C_l, j\in C_r$,
    \begin{equation*}
        \mathbb E[S^g_{ij}] 
        = \mathbb E[e^{4\text{tr}(Q)\tau^g_{ij}}]
    \end{equation*}
    Inserting $\tau^g_{ij}=\tau_{ij}+\Delta^g_{ij}$, and using the definition of the species similarity, Eq. \eqref{eq:GTR_similarity_species},
    \begin{equation*}
        \mathbb E[S^g_{ij}] 
        = \mathbb E[e^{4\text{tr}(Q)(\tau_{ij}+\Delta^g_{ij})}]
        =S_{ij} \mathbb E[e^{4\text{tr}(Q)\Delta^g_{ij}}]
        =s_1 \mathbb E[e^{4\text{tr}(Q)\Delta^g_{ij}}].
    \end{equation*}
    It remains to claim that $\mathbb E[e^{4\text{tr}(Q)\Delta^g_{ij}}]$ is constant. 
    Indeed, by the MSC model, $\{\Delta^g_{ij}\}$ are identically distributed for all $i\in C_l, j\in C_r$.
\end{proof}

\subsection{Proof of Lemma \ref{lem:concent_bound}}\label{subsec:conct_ineq_proof}
We first present four auxiliary lemmas, used in the proof of Lemma~\ref{lem:concent_bound}. 
The first lemma is the Matrix Bernstein Inequality, a fundamental tool in bounding the deviation of random matrices, see \cite{berenstein_ineq}.
The proofs of the other three lemmas appear in Appendix~\ref{subsec:proof_auxiliary}.

\begin{lemma}[Matrix Bernstein Inequality]\label{lem:beren_ineq}
    Let $\boldsymbol{A} = \sum_{l=1}^{K} A_{l}$ where $A_{l}$ are independent random symmetric $m \times m$ matrices with $\mathbb{E}[A_{l}] = 0$ and $\|A_{l}\| \leq L$. Define
    \begin{equation}\label{eq:variance_proxy}
    \nu(\boldsymbol{A}) = \left\| \sum_{l=1}^{K} \mathbb{E}[A_{l}^{2}] \right\|.    
    \end{equation}
    Then for all $t \geq 0$,
    \begin{equation}\label{eq:bern_ineq}
    \mathbb{P}\left( \| \boldsymbol{A} \| \geq t \right) \leq 2m \exp\left( -\frac{t^{2}/2}{\nu(\boldsymbol{A}) + Lt/3} \right).
    \end{equation}
\end{lemma} 

The second lemma describes a property of the entries in the gene similarity matrix $S^g$. 
\begin{lemma}\label{lem:entry_noise}
    Assuming $\T^g$ is a random gene tree generated according to the MSC+GTR model, and let $S^g$ be its similarity. 
    Then $\forall i,j\quad 0\leq S^g_{ij}\leq S_{ij}$.
\end{lemma}
 
The third lemma describes a relation between the spectral norms of matrices with non-negative entries. 
\begin{lemma}\label{lem:nonnegetive_bnd}
    Let \( A,B\in \mathbb{R}^{m\times m} \) be matrices whose entries are all non-negative. Then, 
    \begin{align*}
        \|A\| &\leq\|{A+B}\|.
    \end{align*}
\end{lemma}

Let \( |A| \) denote the matrix obtained by taking the entrywise absolute values of \( A \).
The fourth lemma is a basic result concerning the spectral norm of \( |A| \).
\begin{lemma}\label{lem:entrywise_absolute}
    Let \( A\in \mathbb{R}^{m\times m} \) be a matrix. Then, $\|A\|\leq \||A|\|$.
\end{lemma}

With those lemmas established, we present the proof. In what follows, we use the subscript $A_{ij}$ to denote the $(i,j)$-th entry of a matrix $A$.
\begin{proof}[Proof of Lemma~\ref{lem:concent_bound}]
    Multiplying by $K$, we rewrite the probability of interest  \eqref{eq:concentration_bar_S} as follows
    \begin{equation}\label{eq:prob_intrest}
    \mathbb{P}\left( \left\| \bar{S} - \mathbb{E}[S^g] \right\| > t \right) = \mathbb{P}\left( \left\| \sum_{g=1}^{K} \left( S^g - \mathbb{E}[S^g] \right) \right\| > Kt \right). 
    \end{equation}
    Let \( \boldsymbol{S} = \sum_{g=1}^{K} \left( S^g - \mathbb{E}[S^g] \right) \). Since the matrices \( \{ S^g - \mathbb{E}[S^g] \}_{g=1}^K \) are independent and have zero mean, we can apply the Matrix Bernstein Inequality, 
    Lemma \ref{lem:beren_ineq} with \( \boldsymbol{A}=\boldsymbol{S} \)
    and $t$ replaced by $Kt$. 
    Eq. \eqref{eq:bern_ineq} of this lemma depends on two quantities
    $L$ and $\nu(\boldsymbol{S})$.
    To complete the proof, it remains to provide upper bounds for these two quantities.
    
    \paragraph{Upper bound for $L$.}
    Recall that  $L$ is an upper bound on \( \|S^g - \mathbb{E}[S^g]\| \). 
    By Lemma~\ref{lem:entry_noise}, for all $i,j\in [m]$, 
    $S^g_{ij} \in [0,S_{ij}]$. Hence, upon subtracting the expectation, which is also in $[0, S_{ij}]$, 
    \begin{equation}\label{eq:entry_bound}
        |S^g_{ij} - \mathbb{E}[S^g_{ij}]|\leq S_{ij}.
    \end{equation}
    Next, we decompose $S$ as the sum of the following two matrices:
        $$
        S =  |S^g - \mathbb{E}[S^g]| + (S - |S^g - \mathbb{E}[S^g]|), 
        $$ 
    By definition, all entries of the first matrix are non-negative. All entries of the second matrix are non-negative by \eqref{eq:entry_bound}. Hence, applying Lemma~\ref{lem:nonnegetive_bnd} yields, 
    \begin{equation}\label{eq:pre_summand_bound}
        \||S^g - \mathbb{E}[S^g]|\|\leq \|S\|.    
    \end{equation}
    By Lemma~\ref{lem:entrywise_absolute} and Equation~\eqref{eq:pre_summand_bound}, we conclude that
    \begin{equation*}
        \|S^g - \mathbb{E}[S^g] \| \leq \| S\|.     
    \end{equation*}
    Hence, we may take $L=\|S\|$. 

    \paragraph{Upper bound for $\nu(\boldsymbol{S})$.} 
    Denote $V = \mathbb{E}[(S^g - \mathbb{E}[S^g])^2]$. By Equation \eqref{eq:variance_proxy}, $\nu(\boldsymbol{S})=K\|V\|$, hence we bound $\|V\|$. Simple computations give us:  
    \begin{equation*}
        |V_{ij}| = \left| \mathbb{E}\left[\left(S^g-\mathbb{E}[S^g]\right)^{2}\right]_{ij} \right|
        = \left|  \sum_{l=1}^{m} \mathbb{E}\left[\left(S^g-\mathbb{E} [S^g]\right)_{il}\left(S^g-\mathbb{E}[S^g]\right)_{jl}\right] \right|.
    \end{equation*}
    By the triangle inequality,
    and the fact that for any two random variables $X,Y$,  $Cov(X,Y) \leq \sigma(X) \sigma(Y)$
    where $\sigma(X)$ is the standard deviation of $X$,  
    \begin{equation*}
        |V_{ij}|\leq 
        \sum_{l=1}^{m} \left| \mathrm{Cov}(S^g_{il}, S^g_{jl}) \right|
        \leq \sum_{l=1}^{m} \sigma(S^g_{il}) \cdot \sigma(S^g_{jl}).
    \end{equation*}
    Recall that by Lemma~\ref{lem:entry_noise}, $S^g_{il} \in [0,S_{il}]$. Hence, by a well-known result for bounded random variables, $ \sigma(S^g_{il}) \leq \frac{1}{2} S_{il} $ (see, e.g.,~\cite[Chapter 2]{boucheron2013concentration}). 
    Therefore,     
    \begin{equation*}\label{eq:V_S2}
        |V_{ij}| \leq \frac{1}{4}\sum_{l=1}^{m} S_{il} \cdot S_{jl}
        = \frac14 (S^2)_{ij}.
    \end{equation*}
    We express $\frac14 S^2$ as the sum of the following two matrices, 
    \[
    \frac{1}{4}S^{2} = |V| + \left(\frac{1}{4}S^{2} - |V|\right).
    \]    
    By similar arguments as above, all entries in both of these matrices are non-negative. Hence, by Lemmas~\ref{lem:nonnegetive_bnd} and~\ref{lem:entrywise_absolute}, it follows that:   
    \[
    \| V \| \leq \| |V| \| \leq \frac{1}{4} \| S^{2} \|.
    \]
    Since \( S \) is symmetric,  \( \| S^2 \| = \| S \|^2 \), together with $\nu(\boldsymbol{S})=K\|V\|$, we conclude that
    \begin{equation*}
        \nu(\boldsymbol{S}) \leq K\| S \|^{2}/4.
    \end{equation*}
     Inserting the above bounds on $L$ and $\nu(\boldsymbol{S})$ into 
 \eqref{eq:bern_ineq} yields
 the required inequality \eqref{eq:concentration_bar_S}. 
\end{proof}

\subsection{Proofs of auxiliary lemmas} \label{subsec:proof_auxiliary}
Recall the definition from Eq. \eqref{eq:similarity} of the similarity for any pair $i,j$ of terminal nodes. Analogously, we extend this definition for any pair of nodes $h,h'$ in the tree. To prove Lemma \ref{lem:spec_norm_similarity_bnd}, it is useful to introduce the following vector with the similarities of all terminal nodes to their root. 

\begin{definition}\label{def:u_vec}
    Let $r$ be the root of a tree $\mathcal T$ with $m$ terminal nodes. We denote by $u\in\mathbb{R}^m$ the vector of similarities from the root to each of its $m$ leaves. Specifically, 
    for each terminal node $i$, 
    \[
    u_i =S(r,i). 
    \]
\end{definition}
Recall that $\xi$ denotes the upper bound on the similarity between adjacent nodes. The following lemma bounds the $l_1$ norm of $u$.
\begin{lemma}\label{lem:u1_bound}
    Let $\mathcal T$ be a rooted binary tree with $m$ leaves, and let $r$ be its root. Then, 
    \begin{equation}\label{eq:u1_bound}
    \|u\|_1 \leq
    \begin{cases}
        1 & \text{if } \xi \leq 0.5, \\
        \frac{1}{\xi}m^{1 + \log_2 \xi} & \text{if } \xi > 0.5.
    \end{cases} 
    \end{equation}
\end{lemma}
The proof of this auxiliary lemma appears below. It follows that of Lemma 4.7 in \cite{snj}, which stated a lower bound on 
$\|u\|_2$. In what follows, we denote by $S$ the $m\times m$ similarity matrix of the tree, and by $S(h,h')$ the similarity between the nodes $h,h'$.

\begin{proof}[Proof of Lemma \ref{lem:spec_norm_similarity_bnd}] 
    We start with a known inequality for the spectral norm 
    \begin{equation}\label{eq:spec_by_l1_linfty}
        \|S\|\leq \sqrt{\|S\|_1\|S\|_\infty} ,    
    \end{equation}
    see e.g \cite[Ch.~5.6, Ex.~21]{horn2012matrix}, where $\|\cdot\|_1,\|\cdot\|_\infty$ are the operator norms induced by the $l_1$ and $l_\infty$ norms,  respectively. Next, note that in general $\|A^T\|_\infty = \|A\|_1$. 
    Since the similarity matrix $S$ is symmetric,    \begin{equation}\label{eq:l1_l_infty}
        \|S\|_1=\|S\|_\infty. 
    \end{equation}
    Recall that the $\ell_\infty$ operator norm of a matrix is given by the maximum $\ell_1$ norm of its rows, that is, $\|S\|_\infty = \max_i |S_{[i,:]}|_1$, where $S_{[i,:]}$ is the $i$-th row of $S$. 
    Combining this with Eqs.~\eqref{eq:spec_by_l1_linfty},~\eqref{eq:l1_l_infty},  
    \begin{equation} \label{eq:spect_norm_by_row}
        \|S\|\leq \max_i \|S_{[i,:]}\|_1
    \end{equation}
    
    We now bound $\|S_{[i,:]}\|_1$. Let $r$ be the node attached to the $i$-th leaf. We denote $\mathcal T_r:=\mathcal T \backslash \{i\}$ the tree obtained by removing the leaf $i$, and considering the node $r$ as its root. 
    We let $u\in \mathbb{R}^{m-1}$ be the vector of similarities of the remaining $m-1$ leafs to the node $r$, see Definition \ref{def:u_vec}. 
    Using the fact that the similarity function is multiplicative along paths and that $S(i, r) \leq \xi$ by the assumption in the lemma, it follows that for all $j\neq i$,
    \[
    S(i,j)=S(i,r)\cdot S(r,j)\leq \xi\cdot u_j.
    \]
    The row vector $S_{[i,:]}$ contains, in addition, the entry $S_{ii}=1$. Therefore,
    \begin{equation}\label{eq:row_by_uvec}
        \|S_{[i,:]}\|_1\leq 1+\xi\cdot\|u\|_1.
    \end{equation}
    Finally, combining 
    Lemma~\ref{lem:u1_bound} with Eqs.~\eqref{eq:row_by_uvec} and ~\eqref{eq:spect_norm_by_row} completes the proof.
\end{proof}

To complete the proof of Lemma~\ref{lem:spec_norm_similarity_bnd}, we proceed with the proof of Lemma~\ref{lem:u1_bound}.

\begin{proof}[Proof of Lemma~\ref{lem:u1_bound}.] 
    Given a rooted tree $\mathcal{T}$, let $u \in \mathbb{R}^m$ denote the vector of similarities from each 
    of its $m$ terminal nodes to the root $r$, see Definition~\ref{def:u_vec}. 
    For each leaf $i$, we define its 
    depth $d_i$, as the number of edges on the path from $i$ to the root $r$.

    By the multiplicative property of the similarity measure, $u_i$ equals the product of the edge similarities along the path from the root to the leaf $i$. 
    Given the assumption that for all pairs $i,j$ of adjacent nodes, $S(i,j)\leq \xi$, it follows that $u_i\leq \xi^{d_i}$. 
    Hence,     
    \begin{equation}\label{eq:u1_xi_di}
        \|u\|_1 \leq \sum_{i=1}^m \xi^{d_i}=:B
    \end{equation}
    To prove the lemma, we thus need to bound the quantity $B$ on the RHS of  \eqref{eq:u1_xi_di}. Note that $B$ depends on both $\xi$ and on the topology of the tree. 
    To this end, for a given value of $\xi$, we identify the topology with the maximal value of $B$, as we describe now. For any tree, we define a modification in its topology. We find a topology whose value $B$ can not be improved, and therefore it's a local maximum. We show that any other topology can be improved by the mentioned modification, thus the local maximum is unique, namely is global. 
    
    Specifically, given a candidate tree $\cal T$, we consider the next modification to the tree structure: We remove a pair of two terminal nodes $i, j$ with a common neighbor, at depth $d_i=d_j$, and reattach them to a terminal node $k$ at depth $d_k$. What we show next is that for $\xi\leq 0.5$ this change increases $B$ for moving nodes from "shallow" parent node to a new "deeper" parent node, which eventually leads to a caterpillar tree. For $\xi >0.5$, the opposite is true, leading to a balanced tree that maximizes $B$.
    
    Precisely, let $\tilde {\cal T}$ be the modified tree and let $\tilde{B}$ be its corresponding bound. The difference between $\tilde {B}$ to $B$ is equal to
    \[
         \tilde B-B=2\xi^{d_k+1}-2\xi^{d_i}+\xi^{d_i-1}-\xi^{d_k}. 
    \]

    The first two terms are due to the shift of $i,j$ from depth $d_i$ to depth $d_k + 1$. The last two terms are due to the non-terminal node attached to $i,j$ becoming terminal, while $k$ becomes non-terminal. We can rewrite the last equation as:
    \begin{equation}\label{eq:u_s_difference}
        \tilde {B}-B= \xi^{d_i-1}(1-2\xi) - \xi^{d_k}(1-2\xi) = (1-2\xi)(\xi^{d_i-1} - \xi^{d_k})
    \end{equation}

    One can see that Eq.~\eqref{eq:u_s_difference} changes its behavior at $\xi=0.5$. 
    For $\xi \leq 0.5$, the expression in \eqref{eq:u_s_difference} is positive if $d_i - 1 < d_k$. Hence, a new tree $\tilde{\mathcal T}$ where the original pair of adjacent leaves of depth $d_i$ is now at depth $d_k+1$, where $d_k+1 > d_i$, has a larger value $\tilde{B}$. Repeating this step increases the norm until such changes are no longer possible. The extreme case is the tree that contains exactly one terminal node at depth $d_l$ for $d_l = 1,\dots, m - 2$, and two terminal nodes at depth $m - 1$, that is, the caterpillar tree.

    For $\xi > 0.5$, the argument is analogous. In this case, the expression in \eqref{eq:u_s_difference} is positive if $d_i - 1 > d_k$. Thus, moving $x_1,x_2$ from depth $d_i$ to depth $d_k+1$ increases $B$. As in the previous case, we repeatedly apply this modification until we reach the extremal topology - in this case, the perfectly balanced tree, when $m=2^n$ for some $n\in\mathbb{N}$. If $m \neq 2^n$, then the extremal topology can be described as follows. Define the \textit{$l$-th level} of a binary rooted tree $\mathcal{T}$ to be the set of nodes at depth $l$. A level $l$ is said to be \textit{full} if it contains exactly $2^l$ nodes. With this terminology in place, the extremal topology is a binary tree in which all levels up to $\lfloor \log_2 m \rfloor$ are full, and the remaining leaves are placed at depth $\lfloor \log_2 m \rfloor + 1$.

    We now compute $B$ for each of the two extreme cases. For the caterpillar tree, 
    \[
    B = \sum_{l=1}^{m-1} \xi^{l} + \xi^{m-1} = \frac\xi{1-\xi} \left(1 - \xi^{m-1}\right) + \xi^{m-1}.
    \]
    For $\xi \leq 0.5$, since $\xi/(1-\xi)\leq 1$ then
    \[
    \frac\xi{1-\xi} \left(1 - \xi^{m-1}\right) + \xi^{m-1} \leq 1 - \xi^{m-1} + \xi^{m-1} = 1.
    \]
    
    For the balanced tree, let $N$ be the number of nodes at the non-full level, which is $\lfloor \log_2m\rfloor +1$, e.g, for $m=2^n$, $N$ equals zero. Hence
    \[
    B = N \xi^{\lfloor \log_2m\rfloor +1} + (m-N)\xi^{\lfloor \log_2m\rfloor} = 
    (N \xi+ (m-N))\xi^{\lfloor \log_2m\rfloor} = 
    (N \xi+ (m-N))\frac{1}{\xi}\xi^{\lfloor \log_2m\rfloor+1}
    \]
    Since, $\lfloor \log_2m\rfloor+1>\log m $ and $\xi <1$ we can bound from above as follows,
    \[
    B \leq \frac{m}{\xi} \xi^{\log_2 m} = \frac{1}{\xi} m^{1 + \log_{2} \xi},
    \]
    where the second equality follows from basic logarithmic identities, 
    \[
    \xi^{\log_2 m} = m^{\log_m (\xi^{\log_2 m})} = m^{\log_2 m \cdot \log_m \xi} 
    = m^{\log_2 m \cdot \frac{\log_2 \xi}{\log_2 m}} = m^{\log_{2} \xi}.
    \]
    This concludes the proof. 
\end{proof}

\begin{proof}[Proof of Lemma~\ref{lem:entry_noise}]
    Using the expression for $S^g_{ij}$ in Eq. \eqref{eq:GTR_similarity}, the definition of $\Delta^g_{ij}$, and $S_{ij}$,
    \[
    S^g_{ij} 
    = \exp\left(4\text{tr}(Q)\tau^g_{ij}\right) 
    = \exp\left(4\text{tr}(Q)(\tau_{ij}+\Delta^g_{ij})\right)
    = S_{ij}\exp(4\text{tr}(Q)\Delta^g_{ij}).
    \]
    Since $\Delta_{ij}^g$ is non negative random variable, and $\text{tr}(Q)<0$, it follows that $\exp(4\text{tr}(Q)\Delta^g_{ij})$ is supported on [0,1].
\end{proof}

Recall that for a vector $v$ or a matrix $A$, $|v|$ and $|A|$ denote the vector or matrix obtained by taking the absolute value entry-wise.
\begin{proof}[Proof of Lemma \ref{lem:nonnegetive_bnd}]
    Let \( v \in \mathbb{R}^m \) be a unit vector such that \( \|A\| = \|Av\| \). Since \( A \) has non-negative entries, we may assume that all entries of \( v \) are non-negative. Indeed, if some entries of \( v \) are negative, define a new vector \( u \in \mathbb{R}^m \) by \( u_i = |v_i| \). Clearly, \( \|u\| = \|v\| = 1 \), and since \( A \geq 0 \) entrywise, it follows that \( Au \geq |Av|\geq0 \) entrywise, and hence \( \|Au\| \geq \|Av\| \). If the inequality is strict, this contradicts the maximality of \( \|Av\| \); otherwise, we may replace \( v \) with \( u \).
    
    Next, observe that since both \( A \) and \( B \) are entrywise non-negative, the vectors \( Av \) and \( Bv \) are also non-negative. Hence:
    \[
    \|A\| = \|Av\| \leq \|Av + Bv\| = \|(A + B)v\| \leq \|A + B\|.
    \]
\end{proof}

\begin{proof}[Proof of Lemma \ref{lem:entrywise_absolute}]
    Let $A$ be an $m \times m$ matrix, and let $v \in \mathbb R^m, u\in\mathbb R^m$ be defined as in the proof of Lemma \ref{lem:nonnegetive_bnd} above. 
    Using the definition of the spectral norm, of the vector $v$, and the fact that $\|w\|=\||w|\|$ for any vector $w$,
    \begin{equation}\label{eq:A_|A|_1}
        \|A\|
        =\|Av\|
        =\||Av|\|
    \end{equation}
    Next, we turn to upper bound the $i$th entry of $|Av|$. 
    By the triangle inequality,
    \begin{equation*}
        \bigl|(Av)_i\bigr|
        = \left|\sum_{j=1}^m A_{ij} v_j\right| 
        \le \sum_{j=1}^m |A_{ij}|\,|v_j|
        =(|A|u)_i.
    \end{equation*}
    Hence,
    \begin{equation}\label{eq:A_|A|_2}
        \||Av|\|
        \leq \||A|u\|
        \leq \| |A| \|.
    \end{equation}
    Combining Eqs. \eqref{eq:A_|A|_1} and \eqref{eq:A_|A|_2}, completes the proof.
\end{proof}

\end{document}